\documentclass[twoside]{article}

\usepackage[accepted]{arxiv_style}
\usepackage[hyphens]{url}
\usepackage{tikz}
\usepackage{amsthm,amsfonts,amsmath,amssymb,bm}
\usepackage{algorithm,algpseudocode}
\usepackage{arydshln}
\usepackage{mathtools}
\usepackage[hidelinks]{hyperref}
\usepackage{cleveref}
\usepackage{cite}
\usepackage{wrapfig}
\usepackage{bbm}
\usepackage{xfrac}
\usepackage{subcaption}
\usepackage{enumerate}
\usepackage[normalem]{ulem}
\usepackage{balance}
\usepackage{graphicx}

\theoremstyle{plain}
\newtheorem{theorem}{Theorem}
\newtheorem{lemma}{Lemma}

\theoremstyle{remark}

\theoremstyle{definition}
\newtheorem{definition}{Definition}

\newcommand{\testmatrix}{\mathbf{G}}

\newcommand{\znoiseProb}{\mathit{z}}

\renewcommand{\Pr}{\mbox{${\mathbb P}$} }

\newcommand{\edgeSet}{\mathit{E}}
\newcommand{\vertexIndex}{\mathit{v}}

\newcommand{\setItems}{\mathcal{N}}

\newcommand{\nItems}{\mathit{n}}
\newcommand{\nFamilies}{\mathit{F}}
\newcommand{\familyIndex}{\mathit{j}}
\newcommand{\familyInfectRate}{\mathit{q}}
\newcommand{\nDef}{\mathit{k}}
\newcommand{\sparseRegimePar}{\mathit{\alpha}}
\newcommand{\nDefFamilies}{\mathit{k_f}}
\newcommand{\sparseRegimeFamilyPar}{\mathit{\alpha_f}}
\newcommand{\nMembersSymmetric}{\mathit{M}}
\newcommand{\nDefMembersSymmetric}{\mathit{k_m}}
\newcommand{\memberInfectRateSymmetric}{\mathit{p}}
\newcommand{\nMembers}{\mathit{M_\familyIndex}}
\newcommand{\nDefMembers}{\mathit{k_m^\familyIndex}}
\newcommand{\memberInfectRate}{\mathit{p_\familyIndex}}

\newcommand{\defFamilyVariables}{\mathbf{V}}

\newcommand{\sampleSet}{\mathit{r}}
\newcommand{\nSamples}{R}
\newcommand{\mixedSample}{\mathit{x}}
\newcommand{\fracHeavInfectedComb}{\mathit{\phi_{c}}}
\newcommand{\fracHeavInfectedProb}{\mathit{\phi_{p}}}
\newcommand{\createSample}{\mathit{SelectRepresentatives}}
\newcommand{\adapt}{\mathit{AdaptiveTest}}
\newcommand{\nadapt}{\mathit{NonAdaptiveTest}}

\newcommand{\nTests}{\mathit{T}}
\newcommand{\testResult}{\mathit{Y}}
\newcommand{\testIndex}{\mathit{\tau}}

\newcommand{\memberIndex}{\mathit{i}}
\newcommand{\defVector}{\mathbf{U}}
\newcommand{\defVariable}{\mathit{U}}
\newcommand{\defectValue}{\mathit{u}}
\newcommand{\defFamilyVariable}{\mathit{V}}
\newcommand{\defectFamilyValue}{\mathit{v}}

\newcommand{\binEntropy}{\mathit{h_2}}
\newcommand{\entropy}{\mathit{H}}
\newcommand{\Ber}{\text{Ber}}
\newcommand{\boundRegime}{0.38}

\newcommand{\nOnesPerRow}{\mathit{c}}
\newcommand{\nBlocks}{\mathit{b}}
\newcommand\numberthis{\addtocounter{equation}{1}\tag{\theequation}}

\newcommand{\pavlos}[1]{{{#1}}}
\newcommand{\christina}[1]{{#1}}
\newcommand{\SundarComment}[1]{{{#1}}}

\begin{document}
	
%

%
\runningauthor{Pavlos Nikolopoulos, Sundara Rajan Srinivasavaradhan, Tao Guo, Christina Fragouli, Suhas Diggavi}	
	
\twocolumn[
\aistatstitle{Group testing for connected communities}

\aistatsauthor{Pavlos Nikolopoulos$^{\dagger}$ \And Sundara Rajan Srinivasavaradhan$^{\ddagger}$ \And Tao Guo$^{\ddagger}$}
\aistatsauthor{Christina Fragouli$^{\ddagger}$ \And Suhas Diggavi$^{\ddagger}$}

\aistatsaddress{$^{\dagger}$EPFL, Switzerland \And  $^{\ddagger}$University of California Los Angeles, USA}
]


\begin{abstract}
In this paper, we propose algorithms that leverage a known community structure to make group testing more efficient.
We consider a population organized in disjoint communities: each individual participates in a community, and its infection probability depends on the community (s)he participates in. 
Use cases include families, students who  participate in several classes, and workers who share common spaces.
Group testing reduces the number of tests needed to identify the infected individuals by pooling diagnostic samples and testing them together. We show that if we design the testing strategy taking into account the community structure,  we can significantly reduce the number of tests needed for adaptive and non-adaptive group testing, and can improve the reliability in  cases where tests are noisy.
\end{abstract}



\section{Introduction}
Group testing pools together diagnostic samples to reduce the number of tests needed to identify infected members in a population. In particular, if in a population of $\nItems$ members we have a small fraction infected (say $\nDef \ll \nItems$ members),  we can identify the infected members using as low as $\mathcal{O}(\nDef\log(\frac{\nItems}{\nDef}))$ group tests, as opposed to $n$ individual tests \cite{GroupTestingBook,GroupTestingMonograph,kucirka2020-PCR}.
Triggered by the need of widespread testing, such techniques are already being explored in the context of Covid-19~\cite{art1,art2,art4,Cov-GpTest-1,Cov-GpTest-2,kucirka2020-PCR}.
 Group testing has a rich history of several decades dating back to R. Dorfman in 1943   and a number of variations and setups have been examined in the literature \cite{Dorfman,GroupTestingBook,GroupTestingMonograph,nested}.

The observation we make in this paper is that {\em we can leverage a known community
structure to make group testing more efficient.} The work in
group testing we know of, assumes ``independent'' infections, and ignores
that an infection may be governed by community spread; we argue that
taking into account the community structure can lead to significant
savings. As a use case, consider an apartment building consisting of
$\nFamilies$ families that have practiced social distancing; clearly
there is a strong correlation on whether members of the same family  are infected or not. Assume that
the  building management would like to test all members to enable access
to common facilities. We ask, what is the most test-efficient way to do so.

Our approach enlarges the regime where group testing can offer benefits over individual testing. Indeed,
a limitation of group testing is that it offers fewer or no benefits when $\nDef$  grows linearly with  $\nItems$ \cite{LinBndIndv1,LinBndIndv2,LinBndIndv-3-Ungar-1960,individual-optimal,GroupTestingMonograph}. 
Taking into account the community structure allows to identify and remove from the population large groups of infected members, thus reducing their proportion and converting  a linear to a sparse regime identification. 
Essentially, the community structure can guide us on when to use individual, and when group testing.

Our main results are as follows.
Assume that  $\nItems$ population members are partitioned into $\nFamilies$ groups that we call \textit{families}, out of which $\nDefFamilies$ families have at least one infected member.\\ 
$\bullet$ We derive a lower bound on the number of tests, which  for some regimes increases (almost) linearly with $\nDefFamilies$ (the number of infected families) as opposed to $\nDef$ (the number of infected members).\\
$\bullet$ We propose an adaptive algorithm that achieves the lower bound in some parameter regimes. \\
$\bullet$ We propose a nonadaptive algorithm that accounts for the community structure to reduce the number of tests when some false positive errors can be tolerated. \\
$\bullet$ We propose a new decoder based on loopy belief propagation that is generic enough to accommodate any community structure and can be combined with any test matrix (encoder) to achieve low error rates.\\ 
$\bullet$  We numerically show that leveraging the community structure can offer benefits both when the tests used have perfect accuracy and when they are noisy.

We present our models in \Cref{sec:Back-Notn}, the lower bound in \Cref{sec:lower-bounds}, our algorithms for the noiseless case in \Cref{section-algorithm}, 
and loopy belief propagation (LBP) decoding in \Cref{sec:LBP}. 
Numerical results are in \Cref{section-experiments}.

\textbf{Note:} The proofs of our theoretical results (in Sections~\ref{sec:lower-bounds}--\ref{section-algorithm}) are in the Appendix, 
along with an extended explanation of the rationale behind our algorithms. 


\section{Background and notation}
\label{sec:Back-Notn}

\subsection{Traditional group testing}
\label{subsec:TradGP}
Our work extends traditional group testing to infection models that are based on community spread. 
For this reason, we review here known results from prior work. 

Traditional group testing typically assumes a population of $\nItems$ members out of which some are infected. 
Two infection models are considered: (i) in the {\em combinatorial model},  a fixed number of infected members $\nDef$, are selected uniformly at random among all sets of size $\nDef$; (ii) in the {\em probabilistic model}, each item is infected independently of all others with probability $\memberInfectRateSymmetric$, so that the expected number of infected members is $\bar{\nDef} = \nItems\memberInfectRateSymmetric$.
A group test $\testIndex$ takes as input samples from $n_\testIndex$ individuals, pools them together and  outputs a single value:  positive if any one of the samples is infected, and negative if none is infected.  
More precisely, let ${\defVariable_i=1}$  when individual $i$ is infected and $0$ otherwise. 
Then the traditional group testing output $\testResult_\testIndex$ takes a binary value calculated as $\testResult_\testIndex= \bigvee_{i\in \delta_{\testIndex}} \defVariable_i$, 
where $\bigvee$ stands for the \texttt{OR} operator (disjunction) and $\delta_{\testIndex}$ is the group of people participating in the test. 

The performance of a group testing algorithm is measured by the number of group tests $\nTests=\nTests(\nItems)$  needed to identify the infected members (for the probabilistic model, the expected number of tests needed).  
Setups that have been explored in the literature include:\\
$\bullet$ {\em Adaptive vs. non-adaptive testing}: In adaptive testing, we use the outcome of previous tests to decide what tests to perform next. An example of adaptive testing is {\em binary splitting}, which implements a form of binary search.
Non-adaptive testing constructs, in advance, a \emph{test matrix} $\testmatrix\in \{0,1\}^{\nTests\times \nItems}$  where each row corresponds to one test, each column to one member, and the non-zero elements determine the set $\delta_{\testIndex}$.  
Although adaptive testing uses less tests than non-adaptive, non-adaptive testing is  more practical  as all tests can be executed in parallel.  \\
$\bullet$ {\em Scaling regimes of operation}: assume $\nDef=\Theta(\nItems^\alpha)$, we say we operate in the linear regime if $\alpha=1$; in the sparse regime if $0\leq \alpha < 1$; in the very sparse regime if $k$ is constant.


{\bf Known results.}
The following are well established results (see  \cite{CntBnd,GroupTestingBook,GroupTestingMonograph} and references therein):\\
$\bullet$ In the combinatorial model, since $\nTests$ tests allow to distinguish among $2^\nTests$ combinations of test outputs, then to identify all $\nDef$ infected members without error,
we need: $2^\nTests\geq {{\nItems}\choose{\nDef}} \Leftrightarrow \nTests \geq \log_2{\binom{\nItems}{\nDef}}$.
This is known as the {\bf counting  bound}~\cite{CntBnd,GroupTestingBook,GroupTestingMonograph}
and implies that we cannot use less than  $\nTests=O(\nDef\log \frac{\nItems}{\nDef})$ tests.
In the probabilistic model, a similar bound has been derived for the number of tests needed on average: $\nTests \geq \nItems \binEntropy\left(\memberInfectRateSymmetric\right)$,
where $\binEntropy$ is the binary entropy function.\\
$\bullet$  Noiseless adaptive testing can achieve the counting bound for $\nDef=\Theta(\nItems^\alpha)$ and $\alpha \in [0,1)$; 
for non-adaptive testing, this is also true of $\alpha \in [0, 0.409]$, 
if we allow a vanishing (with $\nItems$) error~\cite{GroupTestingMonograph,coja-oghlan19,coja-oghlan20a}.\\
$\bullet$ In the linear regime ($\alpha=1$), group testing offers little benefits over individual testing.
In particular, if the infection rate $\sfrac{\nDef}{\nItems}$ is more than $\boundRegime$,
group testing does not use fewer tests than 1-by-1 (individual) testing
unless high identification-error rates are acceptable~\cite{LinBndIndv1,LinBndIndv2,LinBndIndv-3-Ungar-1960}.

\subsection{Community and infection models}
\label{sec:notation:models}
In this paper, we additionally assume a known community structure: 
the population can be decomposed in $F$ disjoint groups of individuals that we call \textit{families}.
Each family $\familyIndex$ has $\nMembers$ members, so that $\nItems=\sum_{\familyIndex=1}^\nFamilies \nMembers$. 
In the symmetric case, $\nMembers=\nMembersSymmetric$ for all $\familyIndex$ and
$\nItems=\nFamilies \nMembersSymmetric$. 
Note, that the term ``families'' is not limited to real families---we use the same term for any group of people that happen to interact, so that they get infected according to some common infection principle.

We consider the following infection models, that parallel the ones in the traditional setup:\\
$\bullet$ {\bf Combinatorial Model (I).} 
$\nDefFamilies$ of the families are {\em infected}---namely they have at least one infected member. 
The rest of the families have no infected members. 
In each infected family $\familyIndex$, there exist $\nDefMembers$ infected members, 
with  $0\leq \nDefMembers \leq \nMembers$. 
The infected families (resp. infected family members) are chosen uniformly at random out of all families (resp. members of the same family).
For our analysis, we sometimes consider only the symmetric case,
where $\nDefMembers=\nDefMembersSymmetric$ for each family $\familyIndex$. \\
$\bullet$ {\bf Probabilistic Model (II).} A family is infected with probability $\familyInfectRate$
i.i.d. across the families. A member of an infected family~$\familyIndex$ 
is infected, independently from the other members (and other families), with probability $\memberInfectRate > 0$. 
If a family $\familyIndex$ is not infected, then $\memberInfectRate = 0$.\\
When $\nDefMembers=\memberInfectRate \nMembers$ the two models behave similarly.

Our goal is two-fold:  (a) provide new lower bounds for the number of tests $\nTests$ needed to identify all infected members without error;  and  (b)  design  community-aware testing algorithms that are more efficient than traditional  group-testing  ones, in the sense that they can achieve the  same identification accuracy using significantly fewer tests and they can also perform close to the lower bounds in some cases.

\subsection{Noisy testing and error probability} 
\label{subsec:Noise-Error}
In this work we assume that there is no dilution noise, that is, the performance of a test does not depend on the number of samples pooled together. This is a reasonable assumption with genetic RT-PCR tests  where even small amounts of viral nucleotides can be amplified to be detectable \cite{Mullis85,kucirka2020-PCR}. 
However, we do consider noisy tests in our numerical evaluation (\Cref{section-experiments}) using a Z-channel noise model\footnote{In a Z-channel noise model, a test output that should be positive, flips and appears as negative with probability $\znoiseProb$, while a test output that is negative cannot flip.
Thus:
$
\Pr(\testResult_{\testIndex}=1|U_{\delta_{\testIndex}})= \left(\bigvee_{i\in \delta_{\testIndex}} \defVariable_{\memberIndex}\right) (1-\znoiseProb).
$}.
We remark that this is simply a model one may use; 
our algorithms are agnostic to this and can be used with any other model.

Additionally, some of our identification algorithms may return with errors. 
For this, we use the following terminology:
Let $\hat{\defVariable}_\memberIndex$ denote the estimate of the state of $\defVariable_\memberIndex$ after group testing.
{\em Zero error} captures the requirement that $\hat{\defVariable}_\memberIndex=\defVariable_\memberIndex$ for all $\memberIndex\in \setItems$. 
Vanishing error requires that all error probabilities go to zero with~$\nItems$.
Sometimes we also distinguish between {\em False Negative (FN)} and {\em False Positive (FP)} errors: FN errors occur when infected members are identified as non-infected (and vice-versa for FP).

\SundarComment{
\subsection{Other related work}
\christina{The idea of community-aware group testing is explored to some extent in our preprint~\cite{GroupTesting-community}.}
Also, a similar idea of using side-information \pavlos{from contact tracing} in decoding is proposed by~\cite{zhu2020noisy,  goenka2020contact}, independently from our work. 
That work is complementary to ours;
we focus more on test designs rather than decoding, for which we use well-known algorithms such as COMP and LBP. Finally, test designs, lower bounds and decoding algorithms for independent but not identical priors are investigated by \cite{prior}.

The line of work on graph-constrained group testing (see for example \cite{cheraghchi2012graph, karbasi2012sequential, luo2019non})  solves the problem of how to design group tests when there are constraints on which samples can be pooled together, provided in the form of a graph;  in our case, individuals can be pooled together into tests freely.
}

\section{Lower bound on the number of tests}
\label{sec:lower-bounds}
We  compute the minimum number of tests needed to identify all infected members under the zero-error criterion in both community models (I) and (II).

\begin{theorem}[Combinatorial community bound]
	\label{thm:combinatorialBound}
	Consider the combinatorial  model (I) (of \Cref{sec:notation:models}). 
	Any algorithm that identifies all $\nDef$ infected members without error requires a number of tests $\nTests$ satisfying:
	\begin{equation}
	\nTests \ge \log_2{\binom{\nFamilies}{\nDefFamilies}} + \sum_{\familyIndex = 1}^{\nDefFamilies} \log_2{\binom{\nMembers}{\nDefMembers}}.
	\label{eq:combinatorialBound}
	\end{equation}
	For the symmetric case:
	$
	\nTests \ge \log_2{\binom{\nFamilies}{\nDefFamilies}} + \nDefFamilies \log_2{\binom{\nMembersSymmetric}{\nDefMembersSymmetric}}.
	$
\end{theorem}

\noindent\textbf{Observations:}
We make two observations regarding the combinatorial community bound, 
in the case where the number of infected family members follows a ``strongly'' linear regime ($\nDefMembersSymmetric \approx \nMembers$) and the number of infected families $\nDefFamilies$ follows a sparse regime 
(i.e., $\nDefFamilies = \Theta(\nFamilies^\sparseRegimeFamilyPar)$ for $\sparseRegimeFamilyPar \in [0,1)$):

\noindent(a) 
The bound increases almost \textit{linearly} with $\nDefFamilies$ (the number of infected families), 
as opposed to $\nDef$ (the overall number of infected members).
This is because, if the infection regime about families is sparse, 
the following asymptotic equivalence holds:
$
\log_2 {\binom{\nFamilies}{\nDefFamilies}} \sim
\nDefFamilies \log_2 {\frac{\nFamilies}{\nDefFamilies}} \sim (1-\sparseRegimeFamilyPar) \nDefFamilies \log_2 {\nFamilies} \label{eq:asymptoticRelation1}
$.

\noindent(b) 
If additionally to the sparse regime about families, 
an overall sparse regime ($\nDef = \Theta(\nItems^\sparseRegimePar)$ for $\sparseRegimePar \in [0,1)$) holds, 
then the community bound may be significantly lower than the counting bound that does not take into account the community structure.
Consider, for example, the symmetric case. 
The asymptotic behavior of the counting bound in the sparse regime is: 
$\log_2 \binom{\nItems}{\nDef} \sim \nDef \log_2 \frac{\nItems}{\nDef} \sim \nDefFamilies \nDefMembersSymmetric \log_2{\frac{\nFamilies}{\nDefFamilies}}$, where the latter is because $\nDefMembersSymmetric \approx \nMembersSymmetric$.
So, the ratio of the counting bound to the combinatorial bound scales (as $\nFamilies$ gets large) as: 
\begin{align}\label{eq:combinatorialBenefit}
\frac{\log_2 \binom{\nItems}{\nDef}}{\log_2{\binom{\nFamilies}{\nDefFamilies}} + \nDefFamilies \log_2{\binom{\nMembersSymmetric}{\nDefMembersSymmetric}}} \sim 
\frac{\nDefFamilies \nDefMembersSymmetric \log_2 \frac{\nFamilies}{\nDefFamilies}}
{\nDefFamilies\log_2 \frac{\nFamilies}{\nDefFamilies}}= \nDefMembersSymmetric.
\end{align}
Although simplistic, observation (b) is important for practical reasons. 
Many times, the population is composed of a large number of families with members that have close contacts (e.g. relatives, work colleagues, students who attend the same classes, etc.). 
In such cases, we do expect that almost all members of infected families are infected (i.e. $\nDefMembersSymmetric \approx \nMembers$), even though the overall infection regime may still be sparse.
Eq.~\eqref{eq:combinatorialBenefit} shows the benefits of taking the community structure into account in the test design, in such a case. 
\SundarComment{
\begin{theorem}[Probabilistic Community bound]
	\label{thm:probabilisticBound}
	Consider the probabilistic  model (II) (of \Cref{sec:notation:models}).
	Any algorithm that identifies all $\nDef$ infected members without error requires a number of tests $\nTests$ satisfying:
	\begin{align*}
 	    &\nTests \geq \nFamilies \binEntropy(\familyInfectRate)
 	    \pavlos{+\sum_{\familyIndex = 1}^{\nFamilies} \familyInfectRate \nMembers \binEntropy(\memberInfectRate)} 
 	    -
 	    w_\familyIndex 
 	    \binEntropy\left (\frac{1-\familyInfectRate}{w_\familyIndex}\right ) \numberthis
	    \label{eq:probabilisticBound}
 	\end{align*}
	where
	$ w_\familyIndex = 1-\familyInfectRate+\familyInfectRate(1-\memberInfectRate)^{\nMembers}$. 
\end{theorem}}

\textbf{Two observations:} 
(a) \pavlos{If for each family $\familyIndex$, $\memberInfectRate$ and $\nMembers$ are such that $\familyInfectRate(1-\memberInfectRate)^{\nMembers} \rightarrow 0$ (i.e. the probability of the peculiar event, where a family is labeled ``infected'' and yet has no infected members, is negligible),} 
the combinatorial and probabilistic bounds are asymptotically equivalent. 
In particular, using the standard estimates of the binomial coefficient~\cite[Sec.~4.7]{information-theory-ash},
the combinatorial bound in~\eqref{eq:combinatorialBound}
is asymptotically equivalent to 
$\nFamilies \binEntropy(\sfrac{\nDefFamilies}{\nFamilies}) + \sum_{\familyIndex = 1}^{\nDefFamilies} \nMembers \binEntropy(\sfrac{\nDefMembers}{\nMembers})$, 
which matches 
the probabilistic bound in~\eqref{eq:probabilisticBound}:
$\nFamilies \binEntropy(\familyInfectRate) + 
\familyInfectRate \sum_{\familyIndex = 1}^{\nFamilies} \nMembers \binEntropy(\memberInfectRate) =
\nFamilies \binEntropy(\sfrac{\bar{\nDefFamilies}}{\nFamilies}) +
\sum_{\familyIndex = 1}^{\bar{\nDefFamilies}} \nMembers \binEntropy(\sfrac{\bar{\nDef}_m^\familyIndex}{\nMembers})
$,   
with
$\nDefFamilies = \bar{\nDefFamilies} + o(1)$ and
$\nDefMembers = \bar{\nDef}_m^\familyIndex + o(1)$
in place of their expected values $\bar{\nDefFamilies} = \nFamilies \familyInfectRate$ and $\bar{\nDef}_m^\familyIndex$.

\SundarComment{
(b) Theorem~\ref{thm:probabilisticBound} extends from zero-error recovery to constant-probability recovery
by applying Fano’s inequality (similarly to Thm 1 of~\cite{prior}), and in doing so, the right-hand side of \eqref{eq:probabilisticBound} gets multiplied by the desired probability of success $\Pr(suc)$.}

\section{Algorithms}\label{section-algorithm}


\subsection{Adaptive algorithm}
\label{sec:adapt}
\begin{algorithm}
\caption{Adaptive Community Testing}
\label{algorithm-noiseless}
$\hat{\defVariable}_\memberIndex$ is the estimated infection status of member $\memberIndex$.\\
$\hat{\defVariable}_{\mixedSample}$ is the estimated infection status of a mixed sample $\mixedSample$.\\
$\createSample()$ is a function that selects a representative subset from a set of members.\\
$\adapt()$ is an adaptive algorithm that tests a set of items (mixed samples or members).
\begin{algorithmic}[1]
	\For {$j = 1,\ldots,\nFamilies$}
	\State{$\sampleSet_\familyIndex = \createSample\left(\left\{\memberIndex: \memberIndex \in \familyIndex \right\}\right)$}
	\EndFor
	\State{$\left[\hat{\defVariable}_{\mixedSample(\sampleSet_1)},
		\ldots,
		\hat{\defVariable}_{\mixedSample(\sampleSet_\nFamilies)} \right]  = \adapt \left(\mixedSample(\sampleSet_1), \ldots, \mixedSample(\sampleSet_\nFamilies) \right) $}
	\State{Set $A:=\emptyset$}
	\For {$j = 1,\ldots,\nFamilies$}
	\If {$\hat{\defVariable}_{\mixedSample(\sampleSet_\familyIndex)} = $ ``positive''}
	\State{Use a noiseless, individual test for each family member: $\hat{\defVariable}_\memberIndex = \defVariable_\memberIndex$, $\forall \memberIndex \in \familyIndex$.}
	\Else
	\State{$A := A \cup \left\{\memberIndex : \memberIndex \in \familyIndex \right\}$}
	\EndIf
	\EndFor
	\State{$\left\{\hat{\defVariable}_{\memberIndex}: \memberIndex \in A \right\}  = \adapt \left(A\right)$}\\
	\Return $\left[\hat{\defVariable}_1,\ldots,\hat{\defVariable}_\nItems\right]$
\end{algorithmic}
\end{algorithm}

Alg.~\ref{algorithm-noiseless} describes our algorithm for the fully adaptive case, which consists of two parts (the interested reader may find  the detailed rationale for our algorithm in
~\Cref{app:sec:Adaptive}).
In both parts, we make use of a classic adaptive-group-testing algorithm $\adapt()$, which is an abstraction for any existing (or future) adaptive group-testing algorithm. 
We distinguish between $2$ different kinds of input for $\adapt()$:
(a) a set of selected members, 
which is the typical input of group-testing algorithms;
(b) a set of selected \textit{mixed samples}. 
A mixed sample is created by pooling together samples from multiple members that usually have some common characteristic. 
For example, mixed sample $\mixedSample(\sampleSet_\familyIndex)$ denotes an aggregate sample of a set of representative members $\sampleSet_\familyIndex$ from family $\familyIndex$. 
A mixed sample is ``positive,'' 
if at least one of the members that compose it is infected, 
and ``negative'' otherwise. 
Because in some cases we only care about mixed samples, 
we can treat them in the same way as individual samples---hence use group testing to identify the infection state of mixed samples as we do for individuals.

\textbf{Part 1 (lines 1-4):}
\pavlos{The goal of this part is to detect the infection \textit{regime} inside each family $\familyIndex$, so that the family is tested accordingly at the next part: using group testing, if $\familyIndex$ is ``lightly'' infected, or individual testing, otherwise. 
Our idea is motivated by the result presented in \Cref{subsec:TradGP} that group testing is preferable to individual, only if infection rate is low (i.e. $\memberInfectRate \le \boundRegime$).
Therefore, the challenge is to accurately detect the infection regime spending only a limited number of tests.
In this paper, we limited our exploration to using only one mixed sample in this regard, but more sophisticated techniques are also possible, some of which are discussed in \Cref{appdx:sec:optimized-versions}.}

First, a representative subset $\sampleSet_\familyIndex$ of family-$\familyIndex$ members is selected using a sampling function $\createSample()$ (lines 1-3). 
Then, a mixed sample $\mixedSample(\sampleSet_\familyIndex)$ is produced for each  subset $\sampleSet_\familyIndex$, 
and an adaptive group-testing algorithm is performed on  top of all representative mixed samples (line 4). 
If our choice of $\adapt()$ offers exact reconstruction (which is usually the case), 
then: $\hat{\defVariable}_{\mixedSample(\sampleSet_\familyIndex)} = \defVariable_{\mixedSample(\sampleSet_\familyIndex)}$.
 
\textbf{Part 2 (lines 5-13):}
We treat $\hat{\defVariable}_{\mixedSample(\sampleSet_\familyIndex)}$ as an estimate of the infection regime inside family $\familyIndex$:
if $\hat{\defVariable}_{\mixedSample(\sampleSet_\familyIndex)}$ is positive, 
then we consider the family to be heavily infected (i.e $\sfrac{\nDefMembers}{\nMembers}$ or $\memberInfectRate \ge \boundRegime$), 
otherwise lightly infected (i.e. $\sfrac{\nDefMembers}{\nMembers}$ or $\memberInfectRate < \boundRegime$).  
Since group testing performs better than individual testing only in the latter case (section~\ref{subsec:TradGP}),
we use individual testing for each heavily-infected family (lines 7-8),
and  adaptive group testing for all lightly-infected ones (line 13).

{\bf Analysis for the number of tests.}
We now compute the maximum expected number of tests needed by our algorithm to detect the infection status of all members without error. 
For simplicity of notation, we present our results through the symmetric case, where $\nMembers = \nMembersSymmetric$, 
$\nDefMembers = \nDefMembersSymmetric$ (combinatorial case) or $\memberInfectRate = \memberInfectRateSymmetric$ (probabilistic case),
and 
$|\sampleSet_\familyIndex| = \nSamples$ for all families:
Let $\createSample()$ be a simple function that performs uniform (random) sampling without replacement,
and consider $2$ choices for the $\adapt()$ algorithm: 
(i) Hwang's generalized binary splitting algorithm (HGBSA)~\cite{hwang},
which is optimal if the number of infected members of the tested group is known in advance;
and (ii), traditional binary-splitting algorithm (BSA)~\cite{binSplitting},
which performs well, 
even if little is known about the number of infected members.

\begin{lemma}[Expected number of tests - Symmetric combinatorial model]
	\label{lem:expectedTests:combinatorial}
	Consider the choices (i) and (ii) for the $\adapt()$ defined above. 
	Alg.~\ref{algorithm-noiseless} succeeds using a maximum expected number of tests:
	\begin{align}
	\bar{\nTests}_{(i)} \le &
	\nDefFamilies \fracHeavInfectedComb
	\left(\log_2{\frac{\nFamilies}{\nDefFamilies \fracHeavInfectedComb}} + 1 + \nMembersSymmetric \right) \nonumber
	\\
	&+ \nDef \left(1-\fracHeavInfectedComb\right) \left(\log_2 {\frac{\nItems - \nDefFamilies \nMembersSymmetric \fracHeavInfectedComb}{\nDef \left(1-\fracHeavInfectedComb\right)}}+1\right) 
	\label{eq:expectedTests:comb:Hwang}\\	
	\bar{\nTests}_{(ii)} \le &
	\nDefFamilies \fracHeavInfectedComb
	\left(\log_2{\nFamilies} + 1 + \nMembersSymmetric \right) + \nonumber 
	\\ 
	&+
	\nDef \left(1-\fracHeavInfectedComb\right) \left(\log_2 \left(\nItems - \nDefFamilies \nMembersSymmetric \fracHeavInfectedComb\right) + 1\right),
	\label{eq:expectedTests:comb:binary}
	\end{align}
	where the inequalities are because of the worst-case performance of HGBSA and BSA, 
	and $\fracHeavInfectedComb$ is the expected fraction of infected families whose mixed sample is positive: 
	\[
	\fracHeavInfectedComb =
	\begin{cases}
	0 & \text{, if $\nSamples = 0$} \\
	1-\sfrac{\binom{\nMembersSymmetric - \nDefMembersSymmetric}{\nSamples}}{\binom{\nMembersSymmetric}{\nSamples}} & \text{, if $1 \le \nSamples \le \nMembersSymmetric - \nDefMembersSymmetric$} \\
	1 & \text{, if $\nMembersSymmetric - \nDefMembersSymmetric < \nSamples \le \nMembersSymmetric.$}
	\end{cases}
	\] 
\end{lemma}

\begin{lemma}[Expected number of tests - Symmetric probabilistic model]
	\label{lem:expectedTests:probabilistic}
	If Alg.~\ref{algorithm-noiseless} uses BSA in place of $\adapt()$, 
	then it succeeds using a maximum expected number of tests:
	\begin{align}
	\bar{\nTests} \le &
	\nFamilies \familyInfectRate \fracHeavInfectedProb \left(\log_2{\nFamilies} + 1 +\nMembersSymmetric\right)   \\
	 &+
	\nItems \familyInfectRate \memberInfectRateSymmetric \left(1-\fracHeavInfectedProb\right) 
	\left(\log_2 \left(\nItems \left(1 - \familyInfectRate \fracHeavInfectedProb \right)\right) + 1\right),
	\label{eq:expectedTests:prob:binary}
	\end{align}
	where the inequality is due to the worst performance of BSA, 
	and $\fracHeavInfectedProb = 1- \left(1-\memberInfectRateSymmetric\right)^\nSamples$ is the expected fraction of infected families whose mixed sample is positive.
\end{lemma}

Lemmas~\ref{lem:expectedTests:combinatorial} and~\ref{lem:expectedTests:probabilistic} are derived (in
~\Cref{app:sec:Adaptive})
as a repeated application of the performance bounds of HGBSA and BSA: 
if out of $\nItems$ members, $\nDef$ are infected uniformly at random, then  
HGBSA (resp. BSA) achieves exact identification using at most: $\log_2{\binom{\nItems}{\nDef}} + \nDef$ (resp. $\nDef\log_2{\nItems} + \nDef$) tests~\cite{GroupTestingMonograph,capacity-adaptive}.

\noindent\textbf{Observations:} 
(a) If heavily/lightly infected families are detected without errors in Part 1, 
our algorithm can asymptotically achieve (up to a constant)
the lower combinatorial bound of Theorem~\ref{thm:combinatorialBound} in particular community structures.
We show this via 2 examples:

First, consider a sparse regime for families 
(i.e. $\nDefFamilies = \Theta(\nFamilies^\sparseRegimeFamilyPar)$ for $\sparseRegimeFamilyPar \in [0,1)$) 
and a moderately linear regime within each family
(i.e. $\sfrac{\nDefMembersSymmetric}{\nMembersSymmetric} \approx 0.5$).
In this case: \\$\log_2{\binom{\nFamilies}{\nDefFamilies}} \sim \nDefFamilies \log_2({\sfrac{\nFamilies}{\nDefFamilies}})$, $\log_2{\binom{\nMembersSymmetric}{\nDefMembersSymmetric}} \sim \nMembersSymmetric \binEntropy(\sfrac{\nDefMembersSymmetric}{\nMembersSymmetric}) \sim \nMembersSymmetric$
and the bound in~\eqref{eq:combinatorialBound} becomes: $\nDefFamilies \left(\log_2{\sfrac{\nFamilies}{\nDefFamilies}} + \nMembersSymmetric \right)$.
If $\nSamples$ is chosen such that 
all infected families (which are also heavily infected as $\sfrac{\nDefMembersSymmetric}{\nMembersSymmetric} > 0.38$) are detected without errors (e.g. if $\nSamples > \nMembersSymmetric - \nDefMembersSymmetric$),
then
$\fracHeavInfectedComb = 1$;
thus, the RHS of~\eqref{eq:expectedTests:comb:Hwang} becomes almost equal (up to constant $\nDefFamilies$) to the lower bound~\eqref{eq:combinatorialBound}.

Second, consider the opposite example, 
where the infection regime for families is very high, 
while each separate family is lightly infected.
In this case, $\nDef = \nDefFamilies \nDefMembersSymmetric \approx \nDefFamilies$;
therefore, the lower bound becomes: $\nTests \sim \nDefFamilies \log_2({\sfrac{\nFamilies}{\nDefFamilies}}) + \nDefFamilies \nDefMembersSymmetric \log_2(\sfrac{\nMembersSymmetric}{\nDefMembersSymmetric})
\approx \nDef \log_2({\sfrac{\nItems}{\nDef}})$.
If $\nSamples$ is chosen such that 
none of the (lightly infected) families is marked as heavily infected in Part 1 (e.g. if $R=0$\pavlos{, which reduces to using traditional community-agnostic group testing}), 
then $\fracHeavInfectedComb = 0$,
and the RHS of~\eqref{eq:expectedTests:comb:Hwang} is almost equal (up to $\nDef$) to the bound in~\eqref{eq:combinatorialBound}.

(b) The upper bound in~\eqref{eq:expectedTests:comb:binary} shows that   
our algorithm achieves significant benefits compared to classic BSA when the infected families are heavily-infected and $\nSamples$ is chosen such that $\fracHeavInfectedComb = 1$ (e.g. $\nSamples > \nMembersSymmetric - \nDefMembersSymmetric$);
this is because $\bar{\nTests}_{(ii)} \le \nDefFamilies
\left(\log_2{\nFamilies} + 1 + \nMembersSymmetric \right) \ll \nDef \log_2{\nItems} + \nDef$).
Also, it achieves the same performance as BSA, 
when families are lightly-infected and $\nSamples$ is chosen such that $\fracHeavInfectedComb = 0$ (e.g. $\nSamples = 0$);
this is because $\bar{\nTests}_{(ii)} \le \nDef \log_2{\nItems} + \nDef$).
Since the former case (heavy infection) is more realistic, our algorithm is expected to per. 
form a lot better than classic group testing in practice.

\pavlos{The examples in observation (a) and the above analysis indicate two things: 
First, the knowledge of the community structure is more beneficial when families are heavily infected; traditional group testing performs equally well in low infection rates. 
Our experiments showed that the community structure helps whenever $\memberInfectRateSymmetric > 0.15$ and the benefits increase with $\memberInfectRateSymmetric$.
Second, a rough estimate of the families' infection rate has to be known a priori in order to optimally choose $\nSamples$.
In \Cref{app:sec:Adaptive}, we demonstrate that this is unavoidable \christina{in the symmetric scenario we examine and} when only one mixed sample per family is used to identify which families are heavily/lightly infected.
}

\SundarComment{
(c) In the most favorable regime for our community-aware group testing,
where very few families have almost all their members infected
(i.e. $\nDefFamilies = \Theta(\nFamilies^\sparseRegimeFamilyPar)$ for $\sparseRegimeFamilyPar \in [0,1)$
and $\nDefMembersSymmetric \approx \nMembersSymmetric$), 
even if $\nSamples$ is chosen optimally such that $\fracHeavInfectedComb = 1$,
the ratio of \christina{the expected} number of  tests needed by \Cref{algorithm-noiseless} (see  \eqref{eq:expectedTests:comb:Hwang}) and HGBSA cannot be less than $1/\log(\nItems/\nDef)$, which upper bounds the benefits one may get.
In \Cref{appdx:sec:optimized-versions}, we detail this observation and provide an optimized version of our algorithm that improves upon the gain of $1/\log(\nItems/\nDef)$.}

\subsection{Two stage algorithm} \label{section-two-stage}
The adaptive algorithm can be easily implemented as a two-stage algorithm, where we  first perform  one round of tests, see the outcomes, and then design and perform a second round of tests. The first round of tests implements part~1, checking whether a family is highly infected or not; the second round of tests implements part~2, performing individual tests for the members of the highly infected families, and in parallel, group testing for  the members of the remaining families.

As we did before for the adaptive case, we here make use of a classic non-adaptive group-testing algorithm, which we call $\nadapt()$, and  abstracts any existing (or future) non-adaptive algorithm in the group-testing literature. 
Thus to translate Alg.~\ref{algorithm-noiseless} to a two-stage algorithm, lines 4 and 13 simply become:
\begin{align}
&4: \left[\hat{\defVariable}_{\mixedSample(\sampleSet_1)},...,
		\!\hat{\defVariable}_{\mixedSample(\sampleSet_\nFamilies)} \right]
	\!=\!\nadapt \left(\mixedSample(\sampleSet_1),..., \mixedSample(\sampleSet_\nFamilies) \right)  \nonumber\\
&13: \left\{\hat{\defVariable}_{\memberIndex}: \memberIndex \in A \right\}  = \nadapt \left(A\right).  \label{step13}
\end{align}

{\bf Number of tests:} In some regimes, the two-stage algorithm can operate with the same (order) number of tests as the adaptive algorithm, at a cost of a vanishing error probability:  for example, for the tests in line 4, if $\nDefFamilies=\Theta(\nFamilies^\sparseRegimeFamilyPar)$ with
$\sparseRegimeFamilyPar<0.409$, we can use approximately 
$(1-\sparseRegimeFamilyPar) \nFamilies^\sparseRegimeFamilyPar \log_2 {\nFamilies}$
tests and achieve vanishing error probability leveraging literature nonadaptive algorithms \cite{GroupTestingMonograph,PhaseTrans-SODA16,ncc-Johnson,coja-oghlan19,coja-oghlan20a}. 

\subsection{Non-adaptive algorithm} 
\label{sec:nonadapt}
For simplicity of notation, we describe  our non-adaptive algorithm using again the symmetric case. 

{\bf Test Matrix Structure.}
Our test matrix $\testmatrix$ is divided into two sub-matrices: 
$\testmatrix=\left[
\begin{array}{c}
\testmatrix_1\\
\testmatrix_2
\end{array}
\right].  $

\noindent $\triangleright$ The sub-matrix $\testmatrix_1$ of size $T_1\times \nItems$ identifies the infected families using one mixed sample from each family, similar 
to line 4 of Alg.~\ref{algorithm-noiseless}.
We want $\testmatrix_1$ to identify all (non-)infected families with small error probability. 
If the number of tests available is high, 
we set $T_1=\nFamilies$, i.e., we use one row for each family test.
Otherwise, in sparse $\nDefFamilies$ regimes, we set $T_1$ closer to $O(\nDefFamilies\log \frac{\nFamilies}{\nDefFamilies})$.

\noindent $\triangleright$ The sub-matrix $\testmatrix_2$ of size  $T_2\times \nItems$  has a block matrix structure and contains $\nFamilies$ identity matrices $I_\nMembersSymmetric$, one for each family.
$\testmatrix_2$ is designed as follows: 
(i) each block column contains only one identity matrix $I_\nMembersSymmetric$, i.e., each member is tested only once; 
(ii) each block row $i~(i\in\{1,2,\cdots,\nBlocks\})$ contains $\nOnesPerRow_i$ identity matrices $I_\nMembersSymmetric$, i.e., there are $\nOnesPerRow_i$ members included in the corresponding tests.
As a result: $T_2= \nBlocks\nMembersSymmetric$. 
An example with $\nFamilies=6$, $\nBlocks=3$,  $\nOnesPerRow_1=2$, $\nOnesPerRow_2=1$, $\nOnesPerRow_3=3$ is:
\begin{align}
\testmatrix_2=\left[
\arraycolsep=1.4pt\def\arraystretch{1.2}
\begin{array}{cccccc}
I_\nMembersSymmetric & 0_{\nMembersSymmetric\times \nMembersSymmetric}& 0_{\nMembersSymmetric\times \nMembersSymmetric} & I_\nMembersSymmetric & 0_{\nMembersSymmetric\times \nMembersSymmetric} & 0_{\nMembersSymmetric\times \nMembersSymmetric}  \\
0_{\nMembersSymmetric\times \nMembersSymmetric}& I_\nMembersSymmetric & 0_{\nMembersSymmetric\times \nMembersSymmetric} & 0_{\nMembersSymmetric\times \nMembersSymmetric}  & 0_{\nMembersSymmetric\times \nMembersSymmetric} & 0_{\nMembersSymmetric\times \nMembersSymmetric}\\
0_{\nMembersSymmetric\times \nMembersSymmetric} & 0_{\nMembersSymmetric\times \nMembersSymmetric} & I_\nMembersSymmetric  & 0_{\nMembersSymmetric\times \nMembersSymmetric} & I_\nMembersSymmetric & I_\nMembersSymmetric  \\
\end{array}
\right]. \nonumber
\end{align}

\noindent{\bf Decoding.}
From the outcome of the tests in $\testmatrix_1$
we identify the $\nFamilies-\nDefFamilies$ non-infected families, and proceed to remove the corresponding  columns (non-infected members) from $\testmatrix_2$.
We use the remaining columns of $\testmatrix_2$ to identify  infected members according to the rules (which follow the logic of combinatorial orthogonal matching pursuit (COMP) decoding  \cite{Nonadaptive-1,Nonadaptive-2}):\\
 $(i)$	A member is identified as non-infected  if it is included in at least one negative test in $\testmatrix_2$. \\
$(ii)$	All other members, that are only included in positive tests in $\testmatrix_2$, are identified as infected.

{\bf Error Probability.} 
It is perhaps not hard to see that: 
    after the removal of the columns, 
	the block structure of $\testmatrix_2$ helps us obtain a test matrix that is close to an identity matrix -- hence perform ``almost'' individual testing\footnote{An extended analysis about $\testmatrix_2$ is in~\Cref{rationaleG2}.}.
	Also, note that our decoding strategy for $\testmatrix_2$ leads to zero FN errors. 
	Building on these ideas, the following lemmas guide us though a design of $\testmatrix_2$ that minimizes the (FP) error probability.
	
\noindent $\bullet$ {\em Requiring zero-error decoding is too rigid:} the optimal solution is the trivial solution that tests each member individually, but this would require $T_2\geq \nItems$. 

\noindent $\bullet$ {\em The symmetric choice $\nOnesPerRow_i=\nOnesPerRow$ minimizes the error probability.}
As said,
we design $\testmatrix_2$ such that FP errors are minimized. 
A FP may happen if identity matrices $I_\nMembersSymmetric$ corresponding to two or more infected families appear in the same block row of $\testmatrix_2$. 
In this case, some non-infected members may be included in the same test with infected members from other families and identified as infected by mistake.
\begin{lemma} 
	\label{lem:blockRowProb}
	Under models (I) and (II), the probability that there is some block row containing two or more infected families is:
	\begin{align}
	\Pr_{\text{joint}}^{I}&=1-\frac{\sum\limits_{\substack{|\mathcal{B}|=\nDefFamilies:\; \mathcal{B}\subseteq\{1,2,\cdots,\nBlocks\}}} \;\; \prod\limits_{i\in\mathcal{B}}\nOnesPerRow_i}{{\nFamilies \choose \nDefFamilies}},  \label{prob-overlap1} \\
	\Pr_{\text{joint}}^{II}&=1-\prod_{i=1}^\nBlocks\left[(1-\familyInfectRate)^{\nOnesPerRow_i}+\nOnesPerRow_i\familyInfectRate(1-\familyInfectRate)^{\nOnesPerRow_i-1}\right]. \label{prob-overlap2} 
	\end{align}
\end{lemma}

The following lemma offers a test-matrix design that minimizes the system FP probability, defined as:
\begin{equation}
\Pr(\text{any-FP})\triangleq \Pr(\exists i:\hat{\defectValue}_i=1 \text{ and } \defectValue_i=0). \label{system-FP-prob}
\end{equation}

\begin{lemma} \label{lemma_symmetric_c}
The $\Pr(\text{any-FP})$ is minimized for both models (I) and (II), if $\nOnesPerRow_i=\nOnesPerRow$ for all $i\in\{1,\cdots,\nBlocks\}$. 
\end{lemma}
\begin{lemma} \label{lemma_FP-prob}
For $\testmatrix_2$ as in \Cref{lemma_symmetric_c}, the system FP probability for models (I) and (II) equals: 
\begin{align*}
\Pr^{I}(\text{any-FP})&=\left[1-\frac{1}{{\nMembersSymmetric \choose \nDefMembersSymmetric}}\right] \left[1-\frac{{\nTests_2/\nMembersSymmetric \choose \nDefFamilies}(\nFamilies \nMembersSymmetric/\nTests_2)^{\nDefFamilies}}{{\nFamilies \choose \nDefFamilies}}\right]. 
\\
\Pr^{II}(\text{any-FP})&=\left[1-\sum_{i=1}^{\nMembersSymmetric}\left[\memberInfectRateSymmetric^i(1-\memberInfectRateSymmetric)^{\nMembersSymmetric-i}\right]^2 \frac{1}{{\nMembersSymmetric \choose i}}\right] \nonumber \\
&\hspace{-1.1cm} \cdot\left[1-\left((1-\familyInfectRate)^{\frac{\nFamilies\nMembersSymmetric}{\nTests_2}-1}\left(1-\familyInfectRate+\frac{\nFamilies\nMembersSymmetric \familyInfectRate}{\nTests_2}\right)\right)^{\nTests_2/\nMembersSymmetric}\right]. 
\end{align*}
\end{lemma}

$\Pr(\text{any-FP})$ can be pessimistic; a more practical metric is the average fraction of members that are misidentified (error rate):
$
    R(\text{error})\triangleq \sfrac{|\{i:\hat{\defectValue}_i\neq \defectValue_i\}|}{\nItems}.
$

\begin{lemma}
\label{lemma-ErrorRate}
For $\testmatrix_2$ as in \Cref{lemma_symmetric_c}, the error rate is calculated for models (I) and (II) as:
\begin{align}
R_I(\text{error})&< \frac{\nDefFamilies(\nMembersSymmetric-\nDefMembersSymmetric)}{\nFamilies\nMembersSymmetric}\cdot \Pr_{\text{joint}}^{I}, \label{ErrorRate-1} \\
R_{II}(\text{error})&<(1-\memberInfectRateSymmetric)\familyInfectRate\big[1-(1-\familyInfectRate)^{\nOnesPerRow-1}\big]. \label{ErrorRate-2} 
\end{align}
\end{lemma}

%


\section{Loopy belief propagation decoder}
\label{sec:LBP}
We now describe our new algorithm for decoding infection status of the individuals (and families). This is accomplished by estimating the posterior probability  of the corresponding individual  (or family) being infected via \textit{loopy belief propagation} (LBP). 
LBP computes the posterior marginals exactly  when the underlying factor graph describing the joint distribution is a tree (which is rarely the case) \cite{kschischang2001factor}. 
Nevertheless, it is an algorithm of practical importance and has achieved success on a variety of applications.
Also, LBP offers soft information (posterior distributions), which can be proved more useful than hard decisions in the context of disease-spread management.

We use LBP for our probabilistic model, because it is fast and can be easily configured to take into account the community structure \pavlos{leading to more reliable identification}.  
Many inference algorithms exist that estimate the posterior marginals, some of which have also been employed for group testing. 
For example, GAMP~\cite{zhu2020noisy} and Monte-Carlo sampling~\cite{cuturi2020noisy} yield more accurate decoders.
However, taking into account the statistical information provided by the community structure was proved not trivial with such decoders. 
Moreover, the focus of this work is to examine whether benefits from accounting for the community structure (both at the test design and the decoder) exist; 
hence we think that considering a simple (possibly sub-optimal) decoder based on LBP is a good first step; 
we defer more complex designs to future work.

We next describe the factor graph and the belief propagation update rules for our probabilistic model (II).
Let the infection status of each family $\familyIndex$ be $\defFamilyVariable_\familyIndex \sim \Ber(\familyInfectRate)$. 
Moreover, let $\defFamilyVariable(\defVariable_{\memberIndex})$ denote the family that $\defVariable_{\memberIndex}$ belongs to.
\begin{align}\label{eq:factorizedDistribution}
\Pr(\defFamilyVariable_1,&...,\defFamilyVariable_\nFamilies,\defVariable_{1},...,\defVariable_{\nItems},\testResult_{1},...,\testResult_{\nTests}) = \nonumber\\
&\prod_{\familyIndex=1}^\nFamilies \Pr(\defFamilyVariable_\familyIndex) \prod_{\memberIndex=1}^\nItems \Pr(\defVariable_{\memberIndex}|\defFamilyVariable(\defVariable_{\memberIndex})) \prod_{\testIndex = 1}^\nTests \Pr(\testResult_{\testIndex}|\defVariable_{\delta_{\testIndex}}),
\end{align}
where $\delta_{\testIndex}$ is the group of people participating in the test.
\Cref{eq:factorizedDistribution} can be represented by a factor graph, where variable nodes correspond to each random variable  $\defFamilyVariable_\familyIndex,\defVariable_{\memberIndex},\testResult_{\testIndex}$ and factor nodes correspond to $\Pr(\defFamilyVariable_\familyIndex), \Pr(\defVariable_{\memberIndex}|\defFamilyVariable(\defVariable_{\memberIndex})), \Pr(\testResult_{\testIndex}|\defVariable_{\delta_{\testIndex}})$.

Given the result of each test is $y_\testIndex$, i.e., $\testResult_{\testIndex}=y_{\testIndex}$, 
LBP computes the marginals $\Pr(\defFamilyVariable_\familyIndex=v|\testResult_1=y_1,...,\testResult_\nTests=y_T)$  and  $\Pr(\defVariable_{\memberIndex}=u|\testResult_1=y_1,...,\testResult_\nTests=y_\nTests)$, 
by iteratively exchanging messages across the variable and factor nodes.
The messages are viewed as \textit{beliefs} about that variable or distributions (a local estimate of $\Pr(\text{variable}|\text{observations})$). 
Since all random variables are binary, each message is a 2-dimensional vector. 

We use the factor graph framework from \cite{kschischang2001factor} to compute the messages:
Variable nodes $\testResult_\testIndex$ continually transmit the message $[0,1]$ if $Y_\testIndex=1$ and $[1,0]$ if $Y_\testIndex = 0$ on its incident edge, at every iteration.
Each other variable node ($\defFamilyVariable_\familyIndex$ and $\defVariable_{\memberIndex}$) uses the following rule: 
for incident each edge $e$, the node computes the elementwise product of the messages from every other incident edge $e'$ and transmits this along $e$.
For the factor node messages, we derive closed-form expressions for the sum-product update rules (akin to equation (6) in \cite{kschischang2001factor}). 
\pavlos{The exact messages are described in \Cref{app:sec:LBP}.}


\section{Numerical evaluation}\label{section-experiments} 

In this section, we evaluate the benefits (in terms of number of tests and error rate) from taking the community structure into account in practical scenarios, where noiseless or noisy tests are used.

\textbf{Experimental setup I: Symmetric.}
In our simulations, we consider $2$ different use cases about the community structure:
(Community 1) a neighborhood with $\nFamilies = 200$ families of $\nMembersSymmetric = 5$ members each, and
(Community 2) a university department with $\nFamilies = 20$ classes of $\nMembersSymmetric = 50$ students each. 
In each use case, 
we also examine $2$ different infection regimes:
(a) a linear regime, where $\sfrac{\bar{\nDef}}{\nItems} = 0.1$; and
(b) a sparse regime, where $\bar{\nDef} = \sqrt{\nItems} = 32$.
Finally, we consider both noiseless tests that have perfect accuracy and noisy tests that follow the Z-channel  model from \Cref{subsec:Noise-Error}.
For each scenario, 
we average over $500$ randomly generated community structures, in which the members/students are infected according to the symmetric probabilistic model (II):
first a family/class is chosen at random w.p.~$\familyInfectRate$ to be infected and then each of its members/students gets randomly infected w.p.~$\memberInfectRateSymmetric$.

\textbf{Results.} Our results were similar in all scenarios; 
for brevity, we show here only the sparse regime. Further results can be found in the Appendix of the supplementary submitted document.

\begin{figure}
	\vspace{-0.2cm}
	\centering
	\captionsetup{justification=centering}	
	\includegraphics[ height=3cm, width = 0.4\textwidth]{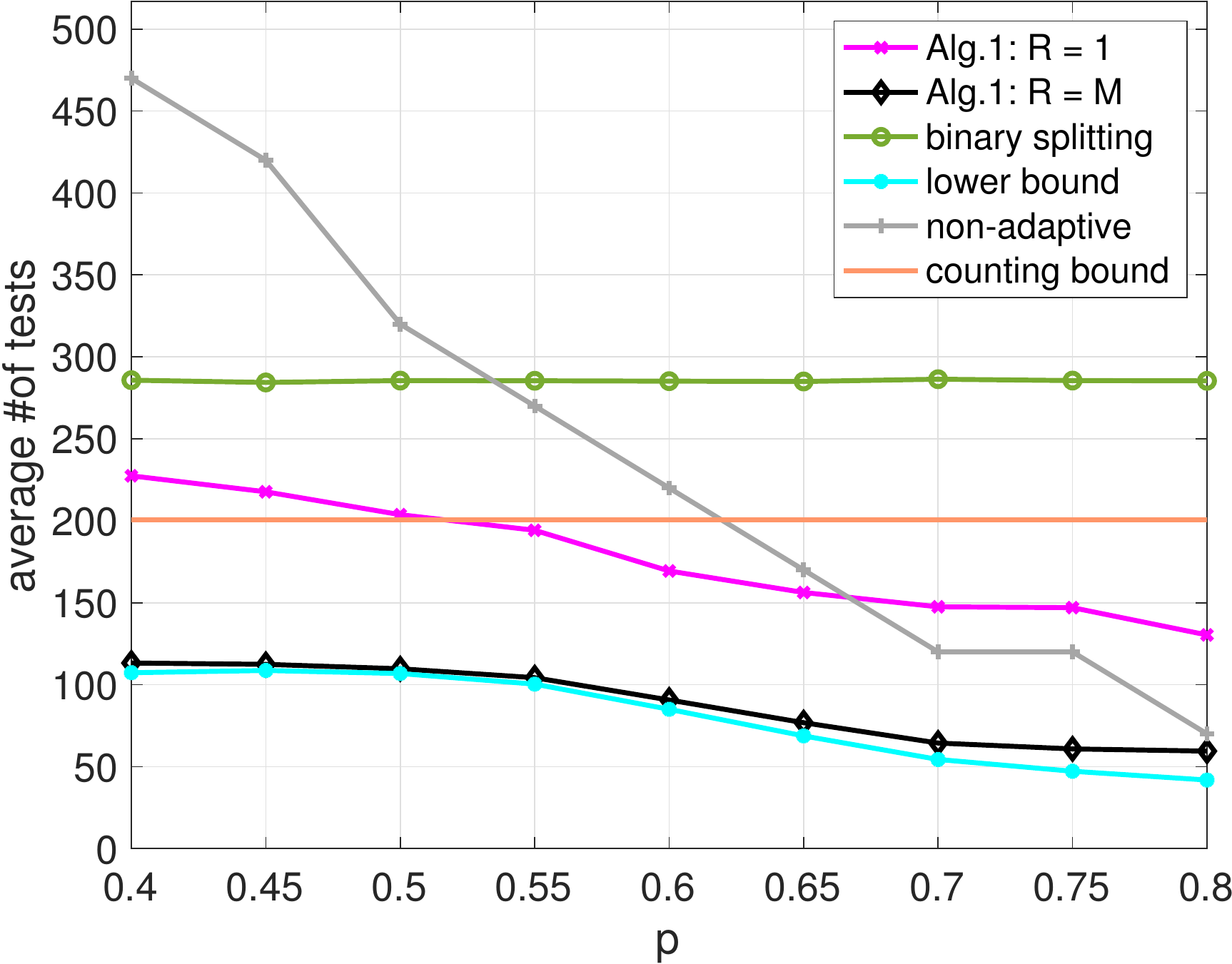}
	\vspace{-0.2cm}
	\caption{Noiseless case: Average number of tests.}
	\label{fig:nTests}
	\vspace{-0.2cm}
\end{figure}

\begin{figure*}[tbh!]
	\begin{minipage}{0.33\textwidth}
	\centering
	\captionsetup{justification=centering}		
	\includegraphics[ height=3cm, width = \textwidth]{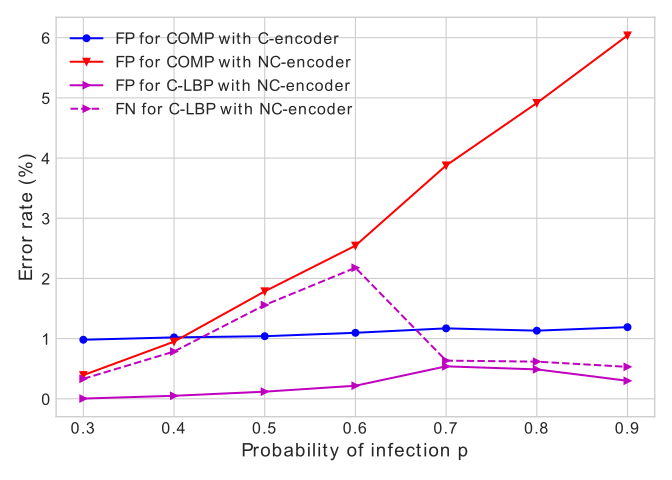}
	\vspace{-0.4cm}
	\caption{Noiseless case: Average error rate with few tests.}
	\label{fig:noiseless-p}
	\end{minipage}\hfill
	\begin{minipage}{0.33\textwidth}
	\centering
	\captionsetup{justification=centering}		
	\includegraphics[ height=3cm, width = \textwidth]{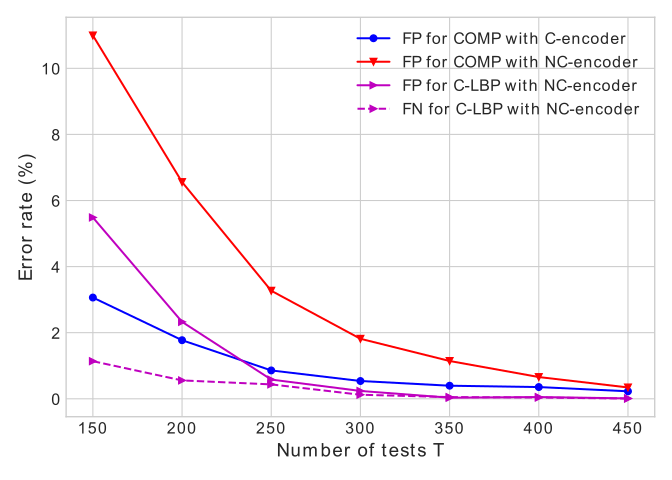}
	\vspace{-0.4cm}
	\caption{Noiseless case: Average error rate ($\memberInfectRateSymmetric = 0.6$).}
	\label{fig:noiseless-T}
	\end{minipage}\hfill
	\begin{minipage}{0.33\textwidth}
	\centering
	\captionsetup{justification=centering}	
	\includegraphics[ height=3cm, width = \textwidth]{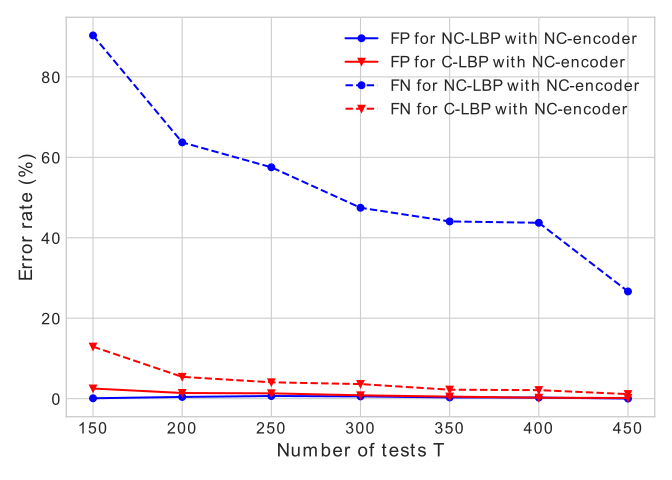}
	\vspace{-0.4cm}
	\caption{Noisy case: Average error rate ($\memberInfectRateSymmetric = 0.8$).}
	\label{fig:noisy}
	\end{minipage}\hfill
\vspace{-0.3cm}
\end{figure*}

(i)~\textit{Noiseless testing -- Average number of tests:}
In this experiment, we measure the average number of tests needed by 3 algorithms that achieve zero-error reconstruction
(Alg.~\ref{algorithm-noiseless} with $\nSamples = 1$, Alg.~\ref{algorithm-noiseless} with $\nSamples = \nMembersSymmetric$,
and
classic BSA),
and a nonadaptive algorithm (\Cref{sec:nonadapt}) that uses $\nTests_1 = \nFamilies$ tests for $\testmatrix_1$ and has FP rate around $0.5\%$.
Alg.~\ref{algorithm-noiseless} assumes no prior knowledge of the number of infected families/classes or members/students, 
hence uses BSA for the $\adapt()$.

Fig.~\ref{fig:nTests} depicts our results about Community 2 and for $\memberInfectRateSymmetric \in [0.4,0.8]$.
Both versions of Alg.~\ref{algorithm-noiseless} need significantly fewer tests compared to classic BSA,
while staying below the counting bound.
This indicates the potential benefits from the community structure,
even when the number of infected members is unknown.
More interestingly, when $\nSamples = \nMembersSymmetric$, Alg.~\ref{algorithm-noiseless}
performs close to the lower bound in most realistic scenarios $\memberInfectRateSymmetric \in [0.5,0.8]$ (as also shown in \Cref{sec:adapt}). 
The relevant result in the linear regime, was slightly worse:~50-70 tests above the lower bound.
Last, the grey line shows number of tests needed by our
nonadaptive algorithm; we observe that even that algorithm can perform better than BSA, when $\memberInfectRateSymmetric > 0.55$ and small FP rates are tolerated. 

(ii)~\textit{Noiseless testing -- Average error rate:}
We here quantify the additional cost in terms of error rate, 
when one goes from a two-stage adaptive algorithm that achieves zero-error identification to much faster single-stage nonadaptive algorithms. 
In each run, we first run 
our two-stage algorithm (\Cref{section-two-stage}) that uses a classic constant-column-weight test design at each stage and measure the number of tests it requires to achieve zero errors.
Then, we use the \textit{same} number of tests to infer the members' infection status through 2 nonadaptive algorithms that account for the community structure either at the test matrix (encoding) part or the decoding and a traditional one that does not consider it at all:
``COMP with C-encoder'' is our nonadaptive algorithm that uses a COMP decoder as described in \Cref{sec:nonadapt};
``C-LBP with NC-encoder'' is an algorithm that uses classic constant-column-weight test design combined with our LBP decoder form \Cref{sec:LBP};
and ``COMP with NC-encoder'' is a traditional nonadaptive algorithm, that we use as a benchmark and uses a constant-column-weight test matrix with a COMP decoder.
``C'' denotes that the community is taken into account, while ``NC'' denotes that it is ignored.
It is important to note that the number of tests needed by the two-stage algorithm (and therefore all other algorithms) gets lower as $\memberInfectRateSymmetric$ gets large, something that affects the results (as discussed further below).

Fig.~\ref{fig:noiseless-p} depicts the FP and FN error rates\footnote{FN rate is the percentage of {\em infected} individuals identified as negative and vice versa for FP.} (averaged over $500$ runs)
as a function of $\memberInfectRateSymmetric \in [0.3,0.9]$ for Community~1. 
We observe that any community-aware nonadaptive algorithm performs better than traditional nonadaptive group testing (red line) when $\memberInfectRateSymmetric > 0.4$---the absolute performance gap ranges from $0.4\%$ (when $\memberInfectRateSymmetric = 0.3$) to $5.5\%$ (when $\memberInfectRateSymmetric = 0.9$).
``COMP with C-encoder'' has a stable FP rate across for all $\memberInfectRateSymmetric$ values that was close to $1\%$, and a zero FN rate by construction.
Our LBP decoder, may yield both FN and FP errors. Also, being an approximate inference algorithm, it may produce worse results than COMP when $\memberInfectRateSymmetric \in [0.42,0.67]$, but performs better when the infection rate is higher. 

Fig.~\ref{fig:noiseless-T} examines the effect of the number of tests.
Starting from the average number of tests used by the two stage algorithm when $\memberInfectRateSymmetric=0.6$, we compute the FP and FN rates for larger numbers of tests.
Our experiment shows a transition around $\nTests = 240$, after which point ``C-LBP with NC-encoder'' performs better than ``COMP with C-encoder''. 
In fact, ``COMP with C-encoder'' seems to converges to zero FP errors much slower. 
This result was common for other $\memberInfectRateSymmetric$ values, the transition just occurred at different $\nTests$. 
We therefore conclude that one may use our ``COMP with C-encoder'' when the number of tests available is limited or they just want to use a simple decoder;
otherwise if the testing budget is larger, one should better go with ``C-LBP with NC-encoder''.

(iii)~\textit{Noisy testing:}
Assuming the Z-channel noise of \Cref{subsec:Noise-Error} with parameter $\znoiseProb = 0.15$, we evaluate the performance of our community-based LBP decoder of \Cref{sec:LBP} against a LBP that does not account for community---namely its factor graph has no $\defFamilyVariable_\familyIndex$ nodes. 

Fig.~\ref{fig:noisy} depicts our results for Community 1 and
for a selected $\memberInfectRateSymmetric = 0.8$ and a number of tests as given from the two-stage algorithm of the previous experiments. 
We observe that the knowledge of the community structure (in C-LBP) reduces both FP and FN rates achieved community-unaware NC-LBP.
Especially, FN error rates drop significantly (up to $80\%$ when tests are few), which is important in our context since FN errors lead to further infections.
Our results were similar for other $\memberInfectRateSymmetric$ values as well.

\textbf{Experimental setup II: Asymmetric.}
\begin{figure}[t!]
	\vspace{-0.2cm}
	\centering
	\captionsetup{justification=centering}	
	\includegraphics[ height=3.5cm, width = 0.35\textwidth]{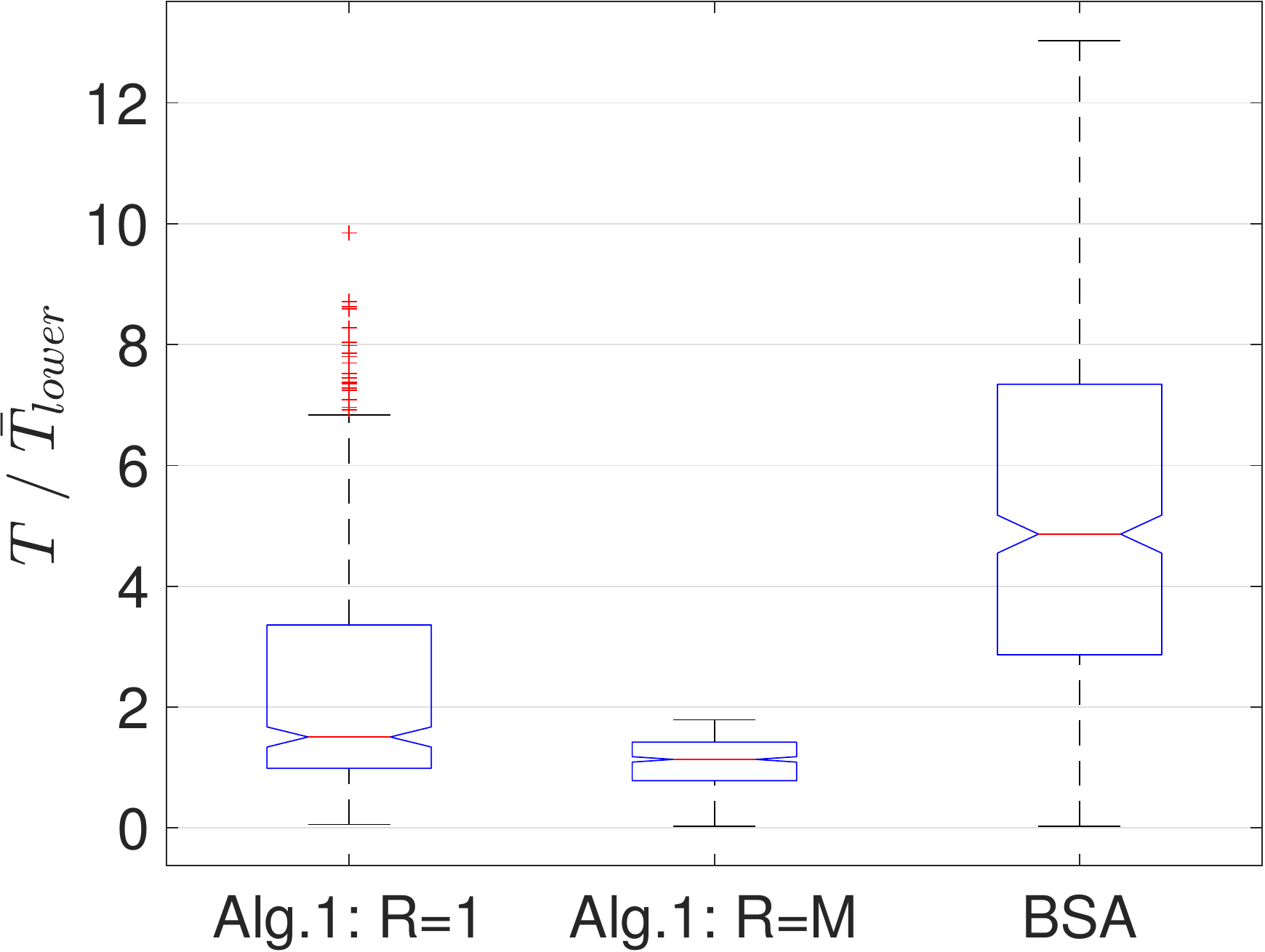}
	\vspace{-0.2cm}
	\caption{Asymmetric case: \SundarComment{Ratio of the number of tests needed to the lower bound \eqref{lem:expectedTests:probabilistic}}.}
	\label{fig:assymteric}
	\vspace{-0.2cm}
\end{figure}
\pavlos{In our asymmetric setup, infections follow 	again the probabilistic model (II), but this time for each family $\familyIndex$, $\nMembers$ and $\memberInfectRate$ are selected uniformly at random from the intervals $[5, 50]$ and $[0.4, 0.8]$, respectively.} 
	
\SundarComment{ Fig.~\ref{fig:assymteric} is a box plot depicting our results for the sparse regime ($\familyInfectRate = 3\%$) over 500 randomly generated instances, as described above. The middle line in the box represents the mean and the ends of the box represent the lower and upper quartiles respectively. The crosses represent outlier points.}
BSA needs on average $5.23\times$ (that can
reach up to $13\times$) more tests compared to the probabilistic bound, while the two versions of \Cref{algorithm-noiseless} with $\nSamples = 1$ and $\nSamples = \nMembersSymmetric$ need only $2.4\times$ and $1.11\times$ (that can reach up to
$9.85\times$ and $1.8\times$) more tests, respectively.
Also, the significantly smaller range between the $25$-th and $75$-th percentiles of the boxplots related to \Cref{algorithm-noiseless} indicate a more predictable performance w.r.t. BSA.

\section{Conclusions}\label{section-conclusions}

\christina{The new observation we make in this paper is that taking into account infection correlations, as dictated by a known community structure,  enables to reduce the number of group tests required to identify the infected members of a population and can improve the identification accuracy when the number of tests is fixed. 

In this paper we make this point assuming a nonoverlaping  community structure, a specific noise model and binary group testing. We considered a combinatorial and probabilistic model, derived lower bounds on the number of tests needed, explored adaptive, two-stage and non-adaptive algorithms for the noiseless case, as well as algorithms for the noisy case. 
	Our algorithms are not always optimal w.r.t.\ the lower bounds, but perform significantly better than community-agnostic group testing; 
	per our experiments, they need upto $55-75\%$ fewer tests (on average) to achieve the same identification accuracy.

We posit that such benefits are possible in a number of other community or noise or group test models; as an example, the followup work in \cite{icc-paper} illustrates benefits when the families overlap. Understanding what are benefits  in more sophisticated models remains as an open question.  }

\clearpage
\pavlos{\section*{Acknowledgments}
This work was supported in part by NSF grants \#2007714, \#1705077 and UC-NL grant LFR-18-548554.
We would also like to thank Katerina Argyraki for her ongoing support and the valuable discussions we have had about this project, as well as the anonymous reviewers (especially Reviewer 1) for their constructive comments that helped improve our paper.} 

\balance
\bibliography{bibliography}

\begin{thebibliography}{}

\bibitem[{Aldridge}, 2019]{individual-optimal}
{Aldridge}, M. (2019).
\newblock Individual testing is optimal for nonadaptive group testing in the
  linear regime.
\newblock {\em IEEE Trans. Inf. Theory}, 65(4).

\bibitem[Aldridge et~al., 2019]{GroupTestingMonograph}
Aldridge, M., Johnson, O., and Scarlett, J. (2019).
\newblock Group testing: an information theory perspective.
\newblock {\em CoRR}, abs/1902.06002.

\bibitem[Ash, 1990]{information-theory-ash}
Ash, R. (1990).
\newblock {\em {Information theory}}.
\newblock Dover Publications Inc., New York, NY.

\bibitem[{Baldassini} et~al., 2013]{capacity-adaptive}
{Baldassini}, L., {Johnson}, O., and {Aldridge}, M. (2013).
\newblock The capacity of adaptive group testing.
\newblock In {\em 2013 IEEE International Symposium on Information Theory},
  pages 2676--2680.

\bibitem[Broadfoot, 2020]{art2}
Broadfoot, M. (2020).
\newblock Coronavirus test shortages trigger a new strategy: Group screening.
\newblock See
  \url{https://www.scientificamerican.com/article/coronavirus-test-shortages-trigger-a-new-strategy-group-screening2/}.

\bibitem[Cai et~al., 2017]{Nonadaptive-2}
Cai, S., Jahangoshahi, M., Bakshi, M., and Jaggi, S. (2017).
\newblock Efficient algorithms for noisy group testing.
\newblock {\em IEEE Trans. Inf. Theory}, 63(4):2113–2136.

\bibitem[Chan et~al., 2014]{Nonadaptive-1}
Chan, C.~L., Jaggi, S., Saligrama, V., and Agnihotri, S. (2014).
\newblock Non-adaptive group testing: Explicit bounds and novel algorithms.
\newblock {\em IEEE Trans. Inf. Theory}, 60(5):3019–3035.

\bibitem[Cheraghchi et~al., 2012]{cheraghchi2012graph}
Cheraghchi, M., Karbasi, A., Mohajer, S., and Saligrama, V. (2012).
\newblock Graph-constrained group testing.
\newblock {\em IEEE Transactions on Information Theory}, 58(1):248--262.

\bibitem[{Coja-Oghlan} et~al., 2020]{coja-oghlan19}
{Coja-Oghlan}, A., {Gebhard}, O., {Hahn-Klimroth}, M., and {Loick}, P. (2020).
\newblock Information-theoretic and algorithmic thresholds for group testing.
\newblock {\em IEEE Trans. Inf. Theory}.

\bibitem[Coja-Oghlan et~al., 2020]{coja-oghlan20a}
Coja-Oghlan, A., Gebhard, O., Hahn-Klimroth, M., and Loick, P. (2020).
\newblock Optimal group testing.
\newblock volume 125 of {\em Proceedings of Machine Learning Research}, pages
  1374--1388.

\bibitem[Cuturi et~al., 2020]{cuturi2020noisy}
Cuturi, M., Teboul, O., and Vert, J.-P. (2020).
\newblock Noisy adaptive group testing using bayesian sequential experimental
  design.
\newblock {\em arXiv preprint arXiv:2004.12508}.

\bibitem[Dorfman, 1943]{Dorfman}
Dorfman, R. (1943).
\newblock The detection of defective members of large population.
\newblock {\em The Annals of Mathematical Statistics}, 14:436--440.

\bibitem[Du and Hwang, 1993]{GroupTestingBook}
Du, D.-Z. and Hwang, F. (1993).
\newblock {\em {Combinatorial Group Testing and Its Applications}}.
\newblock Series on Applied Mathematics.

\bibitem[Ellenberg, 2020]{art4}
Ellenberg, J. (2020).
\newblock Five people. one test. this is how you get there.
\newblock {\em NYtimes}.

\bibitem[Ghosh et~al., 2020]{Cov-GpTest-2}
Ghosh, S. et~al. (2020).
\newblock Tapestry: A single-round smart pooling technique for covid-19
  testing.
\newblock {\em medRxiv}.

\bibitem[Goenka et~al., 2020]{goenka2020contact}
Goenka, R., Cao, S.-J., Wong, C.-W., Rajwade, A., and Baron, D. (2020).
\newblock Contact tracing enhances the efficiency of covid-19 group testing.
\newblock {\em arXiv preprint arXiv:2011.14186}.

\bibitem[Gollier and Gossner, 2020]{art1}
Gollier, C. and Gossner, O. (2020).
\newblock Group testing against covid-19.
\newblock See
  \url{https://www.tse-fr.eu/articles/group-testing-against-covid-19}.

\bibitem[Hu et~al., 1981]{LinBndIndv2}
Hu, M.~C., Hwang, F.~K., and Wang, J.~K. (1981).
\newblock A boundary problem for group testing.
\newblock {\em SIAM Jour. on Algebraic Discrete Methods}.

\bibitem[Hwang, 1972]{hwang}
Hwang, F.~K. (1972).
\newblock A method for detecting all defective members in a population by group
  testing.
\newblock {\em Journal of the American Statistical Association},
  67(339):605--608.

\bibitem[{Johnson} et~al., 2019]{ncc-Johnson}
{Johnson}, O., {Aldridge}, M., and {Scarlett}, J. (2019).
\newblock Performance of group testing algorithms with near-constant tests per
  item.
\newblock {\em IEEE Trans. Inf. Theory}, 65(2):707--723.

\bibitem[Johnson, 2017]{CntBnd}
Johnson, O.~T. (2017).
\newblock Strong converses for group testing from finite block- length results.
\newblock {\em IEEE Trans. Inf. Theory}, 63(9).

\bibitem[Karbasi and Zadimoghaddam, 2012]{karbasi2012sequential}
Karbasi, A. and Zadimoghaddam, M. (2012).
\newblock Sequential group testing with graph constraints.
\newblock In {\em 2012 IEEE information theory workshop}, pages 292--296. Ieee.

\bibitem[Kschischang et~al., 2001]{kschischang2001factor}
Kschischang, F.~R., Frey, B.~J., and Loeliger, H.-A. (2001).
\newblock Factor graphs and the sum-product algorithm.
\newblock {\em IEEE Transactions on information theory}, 47(2):498--519.

\bibitem[Kucirka et~al., 2020]{kucirka2020-PCR}
Kucirka, L.~M., Lauer, S.~A., Laeyendecker, O., Boon, D., and Lessler, J.
  (2020).
\newblock Variation in false-negative rate of reverse transcriptase polymerase
  chain reaction--based sars-cov-2 tests by time since exposure.
\newblock {\em Annals of Internal Medicine}, 173:262--267.

\bibitem[{Li} et~al., 2014]{prior}
{Li}, T., {Chan}, C.~L., {Huang}, W., {Kaced}, T., and {Jaggi}, S. (2014).
\newblock Group testing with prior statistics.
\newblock In {\em 2014 IEEE International Symposium on Information Theory},
  pages 2346--2350.

\bibitem[Luo et~al., 2019]{luo2019non}
Luo, S., Matsuura, Y., Miao, Y., and Shigeno, M. (2019).
\newblock Non-adaptive group testing on graphs with connectivity.
\newblock {\em Journal of Combinatorial Optimization}, 38(1):278--291.

\bibitem[{Nikolopoulos} et~al., 2020]{GroupTesting-community}
{Nikolopoulos}, P., {Srinivasavaradhan}, S.~R., {Guo}, T., {Fragouli}, C., and
  {Diggavi}, S. (2020).
\newblock Community aware group testing.
\newblock {\em arXiv preprint arXiv:2007.08111}.

\bibitem[Nikolopoulos et~al., 2021]{icc-paper}
Nikolopoulos, P., Srinivasavaradhan, S.~R., Guo, T., Fragouli, C., and Diggavi,
  S. (2021).
\newblock Group testing for overlapping communities.
\newblock In {\em Proc. of the IEEE International Conference on Communications,
  {ICC} 2021}.

\bibitem[Riccio and Colbourn, 2000]{LinBndIndv1}
Riccio, L. and Colbourn, C.~J. (2000).
\newblock Sharper bounds in adaptive group testing.
\newblock {\em Taiwanese Journal of Mathematics}, page 669–673.

\bibitem[Saiki et~al., 1985]{Mullis85}
Saiki, R. et~al. (1985).
\newblock Enzymatic amplification of beta-globin genomic sequences and
  restriction site analysis for diagnosis of sickle cell anemia.
\newblock {\em Science}, 230(4732):1350--1354.

\bibitem[Scarlett and Cevher, 2016]{PhaseTrans-SODA16}
Scarlett, J. and Cevher, V. (2016).
\newblock Phase transitions in group testing.
\newblock In {\em Proceedings of the Twenty-Seventh Annual {ACM-SIAM} Symposium
  on Discrete Algorithms, {SODA} 2016}, pages 40--53. {SIAM}.

\bibitem[Sobel and Elashoff, 1975]{sobel1975group}
Sobel, M. and Elashoff, R. (1975).
\newblock Group testing with a new goal, estimation.
\newblock {\em Biometrika}, 62(1):181--193.

\bibitem[{Sobel} and {Groll}, 1959]{binSplitting}
{Sobel}, M. and {Groll}, P.~A. (1959).
\newblock Group testing to eliminate efficiently all defectives in a binomial
  sample.
\newblock {\em The Bell System Technical Journal}, 38(5):1179--1252.

\bibitem[Ungar, 1960]{LinBndIndv-3-Ungar-1960}
Ungar, P. (1960).
\newblock Cutoff points in group testing.
\newblock {\em Comm. Pure Appl. Math}, 13:49–54.

\bibitem[Verdun et~al., 2020]{Cov-GpTest-1}
Verdun, C. et~al. (2020).
\newblock Group testing for sars-cov-2 allows up to 10-fold efficiency increase
  across realistic scenarios and testing strategies.
\newblock {\em medRxiv}.

\bibitem[Walter et~al., 1980]{walter1980estimation}
Walter, S.~D., Hildreth, S.~W., and Beaty, B.~J. (1980).
\newblock Estimation of infection rates in populations of organisms using pools
  of variable size.
\newblock {\em American Journal of Epidemiology}, 112(1):124--128.

\bibitem[Yaakov~Malinovsky, 2016]{nested}
Yaakov~Malinovsky, P. S.~A. (2016).
\newblock Revisiting nested group testing procedures: new results, comparisons,
  and robustness.
\newblock {\em American Statistician}.
\newblock See also \url{https://arxiv.org/abs/1608.06330}.

\bibitem[Zhu et~al., 2020]{zhu2020noisy}
Zhu, J., Rivera, K., and Baron, D. (2020).
\newblock Noisy pooled pcr for virus testing.
\newblock {\em arXiv preprint arXiv:2004.02689}.

\end{thebibliography}

\clearpage
\twocolumn[
\aistatstitle{Appendix}]
\appendix
\numberwithin{equation}{section}


\section{Appendix for \Cref{sec:lower-bounds}: The lower bounds}
\label{app:sec:lower-bounds}
\subsection{Proof of \Cref{thm:combinatorialBound}}
\begin{proof}
	Ineq.~\eqref{eq:combinatorialBound} is because of the following counting argument:
	There are only $2^\nTests$ combinations of test results. 
	But, because of the community model I, there are $\binom{\nFamilies}{\nDefFamilies} \cdot \prod_{\familyIndex = 1}^{\nDefFamilies} \binom{\nMembers}{\nDefMembers}$ possible sets of infected members that each must give a different set of results. Thus, $$2^\nTests \ge \binom{\nFamilies}{\nDefFamilies} \cdot \prod_{\familyIndex = 1}^{\nDefFamilies} \binom{\nMembers}{\nDefMembers},$$ which reveals the result.
	The RHS of the latter inequality is because 
	there are $\binom{\nFamilies}{\nDefFamilies}$ combinations of infected families, 
	and for each infected family $\familyIndex$, there are $\binom{\nMembers}{\nDefMembers}$ possible combinations of infected family members---hence for each combination of $\nDefFamilies$ infected families, 
	there are $\prod_{\familyIndex = 1}^{\nDefFamilies} \binom{\nMembers}{\nDefMembers}$ possible combinations of infected family members.  The symmetric bound is obtained as a corollary by taking $\nMembers = \nMembersSymmetric$ and $\nDefMembers = \nDefMembersSymmetric$ for each infected family $\familyIndex$.
\end{proof}

\subsection{Proof of \Cref{thm:probabilisticBound}}

\begin{proof}
Let 
$\defFamilyVariables$ be the indicator random vector for the infection status of all families.
By rephrasing~\cite[Theorem 1]{prior}, any probabilistic group testing algorithm using $\nTests$ noiseless tests can achieve a zero-error reconstruction of $\defVector$ if:
	\begin{equation}
	\nTests \ge \entropy(\defVector) = \entropy(\defFamilyVariables) + \entropy(\defVector|\defFamilyVariables) - \entropy(\defFamilyVariables|\defVector).
	\end{equation} 	
	The first term is:
	    $\entropy({\defFamilyVariables}) = \sum_{\familyIndex=1}^\nFamilies \entropy(\defFamilyVariable_\familyIndex) = \nFamilies \binEntropy(\familyInfectRate).$
	
	\noindent The second term is calculated as:
	\begin{align*}
	    &\entropy({\bf U}|{\defFamilyVariables}) = \sum_{\vertexIndex=1}^n \entropy(U_v|\defFamilyVariable_{\edgeSet_\vertexIndex}) \\
	     &=\sum_{\vertexIndex=1}^n \sum_{x \in \{0,1\}}\Pr({\defFamilyVariable}_{\edgeSet_\vertexIndex}={x}) \entropy(U_v|\defFamilyVariable_{\edgeSet_\vertexIndex}=x)\\
	    &= \sum_{\vertexIndex=1}^n \left(
	    \familyInfectRate \entropy(U_v|\defFamilyVariable_{\edgeSet_\vertexIndex}=1) + (1-\familyInfectRate) \entropy(U_v|\defFamilyVariable_{\edgeSet_\vertexIndex}=0)\right )\\
	    &= \sum_{\vertexIndex=1}^n 
	    \familyInfectRate \binEntropy(p_{\edgeSet_\vertexIndex}) \pavlos{= \familyInfectRate \sum_{\familyIndex = 1}^{\nFamilies} \nMembers \binEntropy(\memberInfectRate)} 
	    ,
	\end{align*}
	where $\edgeSet_\vertexIndex$ is the family containing vertex $v$.
	
	\noindent Finally, we compute the third term as:
	\begin{align*}
	    &\entropy(\defFamilyVariables|\defVector) = \sum_{\familyIndex=1}^\nFamilies \entropy(\defFamilyVariable_\familyIndex|\defVector)
	    = \sum_{\familyIndex=1}^\nFamilies \entropy(\defFamilyVariable_\familyIndex|\defVector_{\mathcal S_j})\\
	    &= \sum_{\familyIndex=1}^\nFamilies \Pr(\defVector_{S_j}=\mathbf 0) \binEntropy(\Pr(\defFamilyVariable_\familyIndex=0|\defVector_{S_j}=\mathbf 0))\\
	    &= \sum_{\familyIndex=1}^\nFamilies (1-\familyInfectRate+\familyInfectRate(1-\memberInfectRate)^{|S_j|}) \binEntropy\left (\frac{1-\familyInfectRate}{1-\familyInfectRate+\familyInfectRate(1-\memberInfectRate)^{|S_j|}}\right )
	\end{align*}
	where $S_j$ is the set of members who belong to family $\familyIndex$ \pavlos{and $|S_j| = \nMembers$.} 
	Combining all the 3 terms concludes the proof.
\end{proof}

\section{Appendix for \Cref{sec:adapt}: The noiseless adaptive case}
\label{app:sec:Adaptive}
\subsection{Rationale for Alg.~\ref{algorithm-noiseless}}
\pavlos{Group testing already has a rich body of literature with near-optimal test designs in the case of independent infections, we do not try to improve upon them. 
	Instead, we adapt these ideas
	to incorporate the correlations arisen from the community structure.
	All test designs described in this section are conceptually divided into two parts. 
	This split is guided by the community structure and attempts to identify the different infection regimes inside the community, 
	so that the best testing method (individual or classic group testing) is used.
	We show that such a two-part design is enough to significantly reduce the cost of group testing and also achieve the lower bound in some cases.}

\textbf{Two-part design:}
Two parts of \Cref{algorithm-noiseless} serve complimentary goals:

The goal of Part 1 is to detect the infection \textit{regime} inside each family $\familyIndex$: i.e., to accurately estimate  
which of the $\nFamilies$ families have a high infection rate (``heavily'' infected) and which are  have a low or zero infection rate (``lightly'' infected).
Our interest in detecting the infection regime is motivated by prior work~\cite{LinBndIndv1,LinBndIndv2}, 
which has shown that group testing offers benefits over individual testing, only if the infection rate is low ($\memberInfectRate \le \boundRegime$).
This allows us to define the two regimes as follows:
In the combinatorial model I (resp. probabilistic model II), 
a is considered heavily infected when $\sfrac{\nDefMembers}{\nMembers} \ge \boundRegime$ (resp. $\memberInfectRate \ge \boundRegime$);
conversely, it is considered lightly infected family when $\sfrac{\nDefMembers}{\nMembers} < \boundRegime$ (resp. $\memberInfectRate < \boundRegime$).

For each family $\familyIndex$, 
we regard $\hat{\defVariable}_{\mixedSample(\sampleSet_\familyIndex)}$ as an estimate of the family's infection regime.
If $\hat{\defVariable}_{\mixedSample(\sampleSet_\familyIndex)}$ is positive, 
we consider the family to be highly infected and therefore perform individual testing for all of its members. 
Otherwise, if $\hat{\defVariable}_{\mixedSample(\sampleSet_\familyIndex)}$ is negative, we consider the family to be lightly infected and group test its members with all other lightly infected families. 
The challenge is therefore to produce accurate enough regime estimates, 
such that the overall number of tests that are needed from Alg.~$\ref{algorithm-noiseless}$ to achieve exact infection-status reconstruction for all members $\memberIndex = 1,\ldots,\nItems$ is minimal.
We discuss this challenge further below.  

Given all estimates $\hat{\defVariable}_{\mixedSample(\sampleSet_\familyIndex)}$ from Part 1, 
the goal of the Part 2 is then to identify all infected members, 
by using the appropriate testing method (group or individual testing)
according to the infection regime of each family (light or heavy).
In this way, at the end of Part 2, the algorithm returns an estimate $\hat{\defVariable}_\memberIndex$ of the true infection status $\defVariable_\memberIndex$ of each individual member $\memberIndex$.

\begin{figure}
    \centering
    \includegraphics[width = 7cm]{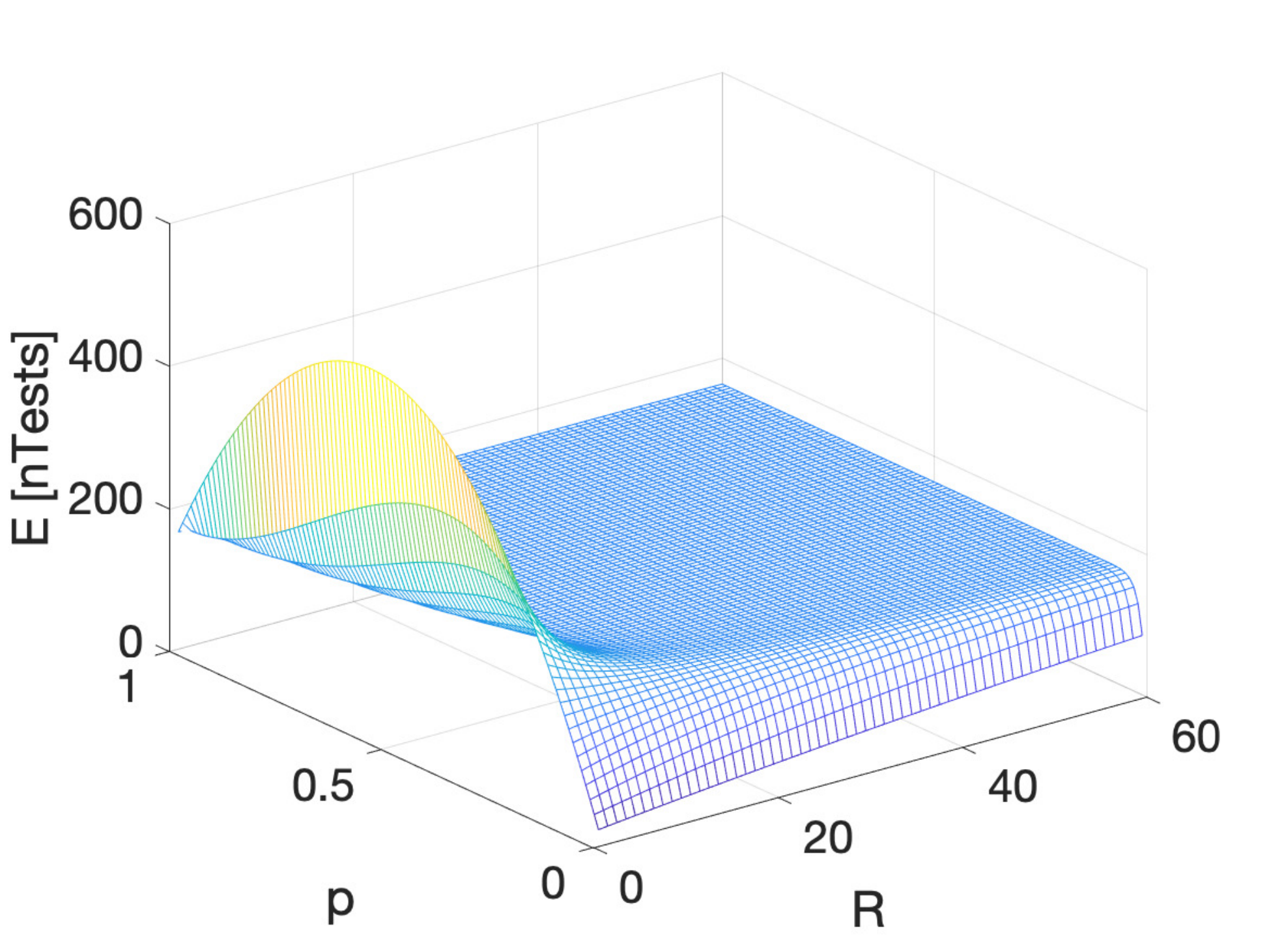}
    \caption{Expected number of tests from \eqref{eq:expectedTests:prob:binary} as a function of size of representative set and probability of infection inside a family.}
    \label{fig:3Dplot}
\end{figure}

\textbf{Selection of family representatives:}
Function $\createSample()$ at line 2 refers to \textit{any} sampling function on a set of family members,  
as long as it returns a fixed number of members from family $\familyIndex$.
That is, one may use their own sampling function, 
as long as the accuracy of Part 1 is well defined. 
In this paper, we consider only random-sampling functions without replacement (i.e. $|\sampleSet_\familyIndex|$ members are randomly chosen from the family members and each subset of that size has the same probability of being selected as the representative subset).
But perhaps, more elaborate sampling functions may be considered 
in other contexts. \christina{For example, if the internal structure of family $\familyIndex$ can be represented through a contact graph,
in which only specific family nodes have external contacts with other families, 
it may make sense to include (some of) these nodes into the representative group with certainty.}

When only one mixed sample per family is used to identify the heavily/lightly infected families, the cardinality of the representative subset $|\sampleSet_\familyIndex|$ is essential, but the optimal choice of it is not trivial.
$|\sampleSet_\familyIndex|$ affects the accuracy of regime estimate---hence the performance of our algorithm in terms of the expected number of tests that it uses.
Unfortunately, choosing the number of representatives optimally is not easy even in the symmetric case that is examined in \Cref{sec:adapt}.
Ideally, in the symmetric case, we would like to choose $|\sampleSet_\familyIndex| = \nSamples$ such that the bounds in Lemmas~\ref{lem:expectedTests:combinatorial} and~\ref{lem:expectedTests:probabilistic} are minimized. 
However, this requires solving equations of the form ${ye^{y}=x}$, 
which is generally possible through Lambert functions for $x \ge -\frac{1}{e}$, but the latter does not hold in our case. 
Fig.~\ref{fig:3Dplot} demonstrates that there exists no unique $\nSamples$ that is optimal for any infection probability $\memberInfectRateSymmetric$ in $(0,1)$ through an example of $\nFamilies = 50$ families with $\nMembersSymmetric = 60$ members each.
The figure plots the bound of Lemma~\ref{lem:expectedTests:probabilistic} as a function of $\memberInfectRateSymmetric$ and $\nSamples$.
As we can see, there is no single minimizer $\nSamples^{\star}$: if $\memberInfectRateSymmetric < 0.15$, then $\nSamples$ must be picked equal to $0$ (which yields traditional group testing); otherwise, if $\memberInfectRateSymmetric > 0.15$, then $\nSamples$ must be selected equal to $\nMembersSymmetric$.

Therefore, in order to optimally choose $\nSamples$, a rough estimate about $\memberInfectRateSymmetric$ has to be known a priori. 
If the latter is not possible, then one may use a few more tests at the first stage of our algorithm to better detect whether a family is heavily infected. We provide such an optimization in the next section.

\textbf{Function $\adapt()$:}
In both parts of our algorithm, we make use of a classic adaptive-group-testing algorithm, which we call $\adapt()$. 
This may be regarded as an abstraction for any existing (or future) adaptive algorithm in the group-testing literature. 
In our analysis, however, we mostly focus on the classic binary splitting algorithm 
because of its good performance in realistic cases, 
where the numbers of infected families and/or members ($\nDefFamilies$, $\nDefMembers$) are unknown~\cite{binSplitting}. 

In this section, we consider only adaptive algorithms that offer noiseless (zero-error) reconstruction.
Note, however, the fact that $\adapt()$ offers exact reconstruction is not enough to guarantee an accurate detection of any family's infection regime in Part 1.
For example, consider the following case, 
where the true infection rate within a family $\familyIndex$ is not very low (say $\memberInfectRate = 0.6$), 
yet none of the family representative in set $\sampleSet_\familyIndex$ happened to be infected. 
Intuitively, the error probability of detection in Part 1 should depend on the number of selected representatives $|\sampleSet_\familyIndex|$ from each family $\familyIndex$
and the infection rate among its members $\memberInfectRate$. 
In our analysis, we examine different scenaria w.r.t.~these parameters and discuss which parametrization (i.e. value of $|\sampleSet_\familyIndex|$) optimizes the expected number of the tests required by our algorithm. 

\subsection{Modified/Optimized versions of Alg.~\ref{algorithm-noiseless}}
\label{appdx:sec:optimized-versions}
\noindent$\bullet$ One modification of our algorithm is the following:
In Part 1, instead of selecting only one representative group for each family, 
we select $m_s$ representative subgroups, each of size $s$, 
and we treat each of these subgroups as a single ``(super)-member''.
That is, we identify whether each subgroup is positive (has at least one positive member) or not, and based on this information, \christina{using for example majority vote,} we can classify the family as heavily or lightly infected; \christina{ essentially we can solve an estimation problem as in \cite{GroupTestingMonograph} (see Chapter 5.3), \cite{walter1980estimation, sobel1975group}}.
In this regard, Alg.~\ref{algorithm-noiseless} is just a special case of this approach, with $m_s=1$ and $s=|\sampleSet_\familyIndex|$.


Intuitively, we expect that such a modification would increase the estimation accuracy of $\hat{\memberInfectRate}$ and reduce the error of the related hypothesis test, at the cost of few more tests.
As a result, it could need fewer tests on expectation than Alg.~\ref{algorithm-noiseless}, 
hence perform better in some cases. 
However, the potential improvement \christina{ would depend on parameters such as the family size - for instance for small size families it {is not expected to be large.}}
To keep things simple, we prefer not to analyze this algorithm in this paper and defer it to future work.

\noindent$\bullet$ Another modification could be the following: 
instead of leveraging the community structure to perform individual tests where needed, 
we could use it to improve traditional binary splitting algorithm by running it on multiple testing groups that are related to the community structure.
For example, consider a symmetric case
where: 
we split all $\nItems = \nFamilies\nMembersSymmetric$ members into $\nMembersSymmetric$ groups of $\nFamilies$ people (one from each family),
then run binary splitting to each of these groups. 

This modification is also related to Hwang's binary splitting algorithm, but achieves only logarithmic benefits compared to binary splitting, as opposed to our algorithm that may perform much better in real cases (see \Cref{sec:adapt}).
In fact, the expected number of tests needed by this modified algorithm would be at most $\nDef \log_2{(\sfrac{\nItems}{\nMembersSymmetric})} + O(\nDef)$:
each group $g$ has $\nDef_g$ infected member and binary splitting needs $\nDef_g \log_2{(\sfrac{\nItems}{\nMembersSymmetric})} + O(\nDef_g)$ tests to identify all of them. By adding together the number of tests for each group $g$, we deduce the result.

\pavlos{\noindent$\bullet$ A last modification occurred to us after a related comment of one of our reviewers, who we thank.
	As discussed in \Cref{sec:adapt}, when a sparse regime holds for families 
	(i.e. $\nDefFamilies = \Theta(\nFamilies^\sparseRegimeFamilyPar)$ for $\sparseRegimeFamilyPar \in [0,1)$) 
	and a heavily linear regime holds within each family
	(i.e. $\nDefMembersSymmetric \approx \nMembersSymmetric$)\footnotetext{The symmetric example is only used here only to better illustrate the advantages of the modification proposed. The idea is similar for the asymmetric case.}, 
	the  benefits of \Cref{lem:expectedTests:combinatorial} with regard to Hwang’s binary splitting (HBSA) cannot be more than $1/\log(\nItems/\nDef)$.
	This is because, in Eq.~\eqref{eq:expectedTests:comb:Hwang}, we get the additive term $\nDefFamilies \nMembersSymmetric > \nFamilies \nMembersSymmetric = \nDef$, which comes from the second stage of \Cref{algorithm-noiseless}.
	
	Nevertheless, if $\nDefMembersSymmetric > \nMembersSymmetric - \nDefMembersSymmetric$ (i.e., the infection rate inside each family is more than $0.5$), then at the second stage of our algorithm it makes more sense to look for not-infected members and stop testing once we find them. 	
	In that case, we need at least $\nDefFamilies * (\nMembersSymmetric - \nDefMembersSymmetric)$ tests, which can be less than k, and therefore could lead to more benefits on average. 
	
	For example, consider the case where $\nDefMembersSymmetric = \nMembersSymmetric-1$. 
	Then the expected number of individual tests needed to find the 1 not-infected member inside each infected family can be computed as follows:
	Without loss of generality, suppose that we test the members at some fixed ordering without replacement and the not-infected member has a uniformly random position in that ordering. Then, the probability of the not-infected item being at a given position $i$ in the ordering is equal to $1/\nMembersSymmetric$ and we need $i$ tests to find it. As a result, the expected number of tests is $\sum_{i=1}^{\nMembersSymmetric} i * 1/\nMembersSymmetric = (\nMembersSymmetric+1)/2$. From linearity of expectation, the expected number of tests for all infected families at the second stage of our algorithm (if we further assume that all infected families are identified without error at the first stage---i.e., $\fracHeavInfectedComb= 1$) will be: $\nDefFamilies * (\nMembersSymmetric+1)/2 < \nDefFamilies * (\nMembersSymmetric-1) = \nDefFamilies*\nDefMembersSymmetric = \nDef$.
	Hence, is this particular regime, the modification of our algorithm can achieve benefits more than $1/\log(\nItems/\nDef)$. 
	
	In the more general case, where $\nMembersSymmetric-\nDefMembersSymmetric > 1$, the relevant probabilities for the computation of the expected number of tests can be obtained from the negative hypergeometric distribution (since sampling is without replacement).
	
	In the extreme case, where for each infected family $\nDefMembersSymmetric$ is known and equal to $\nMembersSymmetric$, all we need to do is to identify the infected families and label all their members as infected. In that case the benefit would be $\nDefFamilies/\nDef$.  
	Note, that to achieve these higher benefits described above, the knowledge of the number of infected members per family is required, but this is also the case for HBSA.	
	
}

\subsection{Proof of \Cref{lem:expectedTests:combinatorial}}
\begin{proof}
	Let $\fracHeavInfectedComb$ be the expected fraction of infected families whose mixed sample is positive. 
	Since $\createSample()$ is uniform random sampling without replacement, we can compute $\fracHeavInfectedComb$ when $1 \le \nSamples \le \nMembersSymmetric - \nDefMembersSymmetric$ using the hypergeometric distribution $Hyper(\nMembersSymmetric,\nDefMembersSymmetric,\nSamples)$, as follows:
	the probability of a random mixed sample $\mixedSample(\sampleSet_\familyIndex)$ being negative (i.e. all members of $\sampleSet_\familyIndex$ are negative) is given by the PMF of $Hyper(\nMembersSymmetric,\nDefMembersSymmetric,\nSamples)$ evaluated at $0$, and it is therefore equal to $\sfrac{\binom{\nMembersSymmetric - \nDefMembersSymmetric}{\nSamples}}{\binom{\nMembersSymmetric}{\nSamples}}$, which yields $\fracHeavInfectedComb = 1-\sfrac{\binom{\nMembersSymmetric - \nDefMembersSymmetric}{\nSamples}}{\binom{\nMembersSymmetric}{\nSamples}}$.
	We also define the following for completeness: $\fracHeavInfectedComb = 0$ when $\nSamples = 0$ and $\fracHeavInfectedComb = 1$ when $\nMembersSymmetric - \nDefMembersSymmetric < \nSamples \le \nMembersSymmetric $.
	
	Fixing the number of positive mixed samples in Part 1 of Alg.~\ref{algorithm-noiseless} to its expected value: $\nDefFamilies \cdot \fracHeavInfectedComb$, we now compute the maximum number of tests needed by the algorithm to succeed.
	
	Alg.~\ref{algorithm-noiseless} performs testing at lines 4, 8, 13.
	
	$\bullet$ At line 4, it identifies the positive mixed samples to mark the corresponding families as heavily infected and all others as lightly infected. 
	If HGBSA is used for $\adapt()$, 
	then Alg.~\ref{algorithm-noiseless} is expected to succeed at this step using $\nDefFamilies \fracHeavInfectedComb \log_2{\frac{\nFamilies}{\nDefFamilies \fracHeavInfectedComb}} + \nDefFamilies \fracHeavInfectedComb$ tests.
	Similarly, if BSA is used for $\adapt()$, 
	then then Alg.~\ref{algorithm-noiseless} is guaranteed to succeed at this step using at most $\nDefFamilies \fracHeavInfectedComb \log_2{\nFamilies} + \nDefFamilies \fracHeavInfectedComb$~\cite{GroupTestingMonograph,capacity-adaptive}.
		
	$\bullet$ At line 8, the expected number of individual tests is equal to:
	$\nMembersSymmetric \nDefFamilies \fracHeavInfectedComb$. 
	This is the same irrespectively from whether $\adapt()$ is binary splitting or Hwang's algorithm as it only depends on $\fracHeavInfectedComb$.
	
	$\bullet$ At line 13, the expected number of items that are tested is:
	$\nItems - \nDefFamilies \fracHeavInfectedComb \nMembersSymmetric$,
	and the expected number of infected members is:
	$\nDef - \nDefFamilies \fracHeavInfectedComb \nDefMembersSymmetric = \nDef \left(1-\fracHeavInfectedComb\right)$.
	So, if HGBSA is used for $\adapt()$, 
	then Alg.~\ref{algorithm-noiseless} is guaranteed to succeed at this step using 
	$\nDef \left(1-\fracHeavInfectedComb\right) \log_2{\frac{(\nItems - \nDefFamilies \fracHeavInfectedComb \nMembersSymmetric)}{\nDef \left(1-\fracHeavInfectedComb\right)}} + \nDef \left(1-\fracHeavInfectedComb\right)$ tests.
	Similarly, if BSA is used, then Alg.~\ref{algorithm-noiseless} is expected to succeed in at most: $\nDef \left(1-\fracHeavInfectedComb\right) \log_2{(\nItems - \nDefFamilies \fracHeavInfectedComb \nMembersSymmetric)} + \nDef \left(1-\fracHeavInfectedComb\right)$ tests~\cite{GroupTestingMonograph,capacity-adaptive}.
	
	We add together all the above terms that are related to HGBSA or BSA, and the result follows. 
\end{proof}
	
\subsection{Proof of \Cref{lem:expectedTests:probabilistic}}
\begin{proof}
	Let $\fracHeavInfectedProb$ be the expected fraction of infected families whose mixed sample is positive. 
	Then, because of the probabilistic setting, 
	$\fracHeavInfectedProb = 1-\left(1-\memberInfectRateSymmetric\right)^{\nSamples}$.
	
	Alg.~\ref{algorithm-noiseless} performs testing at lines 4, 8, 13.
	
	$\bullet$ At line 4, the expected number of mixed samples that are positive is $\nFamilies \familyInfectRate \fracHeavInfectedProb$. 
	So, if BSA is used in the place of $\adapt()$, 
	then the maximum number of tests needed to identify all mixed samples is on expectation $\nFamilies \familyInfectRate \fracHeavInfectedProb \log_2{\nFamilies} + \nFamilies \familyInfectRate \fracHeavInfectedProb$~\cite{GroupTestingMonograph,capacity-adaptive}.
	
	$\bullet$ At line 8, the expected number of individual tests is equal to:
	$\nFamilies \familyInfectRate \fracHeavInfectedProb \nMembersSymmetric$. 
	
	$\bullet$ At line 13, the expected number of items that are tested is:
	$\nItems - \nFamilies \familyInfectRate \fracHeavInfectedProb \nMembersSymmetric$,
	and the expected number of infected members is equal to the expected number of all infected members minus the expected number of the ones that are identified though individual testing at line 8: i.e.,
	$\nFamilies \familyInfectRate \nMembersSymmetric \memberInfectRateSymmetric - \nFamilies \familyInfectRate \fracHeavInfectedProb \nMembersSymmetric \memberInfectRateSymmetric = \nFamilies \familyInfectRate \nMembersSymmetric \memberInfectRateSymmetric \left(1-\fracHeavInfectedProb\right)=
	\nItems \familyInfectRate \memberInfectRateSymmetric \left(1-\fracHeavInfectedProb\right)$.
	So, if BSA is used in the place of $\adapt()$, it is expected to succeed using at most $\nItems \familyInfectRate  \memberInfectRateSymmetric \left(1-\fracHeavInfectedProb\right) \log_2{(\nItems - \nFamilies \familyInfectRate \fracHeavInfectedProb \nMembersSymmetric)} + \nItems \familyInfectRate  \memberInfectRateSymmetric \left(1-\fracHeavInfectedProb\right)$ tests~\cite{GroupTestingMonograph,capacity-adaptive}.
	
	We add together all the above terms and the result follows. 
\end{proof}

\section{Appendix for \Cref{sec:nonadapt}: The Noiseless Non-adaptive case}
\label{app:sec:Non-adaptive}

\subsection{Zero error requirements}    \label{lemma_zero_error_proof}
For our design of $\testmatrix_2$, we have the following lemma and observation.
\begin{lemma} \label{lemma_zero_error}
	To achieve zero-error w.r.t.\ $\testmatrix_2$, we need $\nTests_2\geq \nItems$. 
\end{lemma}
\begin{proof}
	A trivial implementation for $\testmatrix_2$ is to use an identity matrix of size $\nItems$;
	since each member is tested individually, we can identify all the infected members correctly. 
	We next  argue that $\nTests_2\geq \nItems$ for the zero-error case.  We prove this through contradiction.
	Assume that $\nTests_2<\nItems$. Then, from the pigeonhole principle, there exists one member, say ${m}_1$ that does not participate in any test alone -it always participates together with one or more members from a set $\mathcal{S}_1$. Assume that all members in ${S}_1$ are infected, while ${m}_1$ belongs in an infected family but  is not infected -our decoding will result in a FP.
\end{proof}
\noindent{\bf Observation:} $\testmatrix_2$ leads to zero error if and only if it has the following property:\\
{\em Zero Error Property:} Any subset of $\{1,2,\cdots,\nItems\}$ of size $(\nFamilies-\nDefFamilies)\nMembersSymmetric + \nDefFamilies(\nMembersSymmetric-\nDefMembersSymmetric)$ equals the union of some testing rows 
of $\testmatrix_2$.  Namely, the members of the not-infected families together with the not-infected members of the infected families, need to be the only participants in some rows of $\testmatrix_2$, for all possible not-infected families and not-infected members.  This requirement can lead to an alternative proof of \Cref{lemma_zero_error}.

\subsection{Rationale for the structure of $\testmatrix_2$}
\label{rationaleG2}
Our goal  is to design a non-trivial matrix $\testmatrix_2$ that can identify almost all the infected members with high probability and a small number of tests.
We next discuss two intuitive properties we would like our designs to have to minimize the error probability.\\

\noindent{\em Desirable Property 1:} Use identity matrices as building blocks.\\
{\it Intuition:} ideally, after removing the $(\nFamilies-\nDefFamilies)\nMembersSymmetric$ columns corresponding to the members in non-infected families, we would like the remaining columns to form an identity matrix so that we can identify all the infected members correctly. To reduce the number of tests, there should be more than one members included in each test. 
Thus we  use overlapping identity matrices, one corresponding to each family. We assume the index for the $\nItems$ members is family-by-family, i.e., the indexes for the members in the same family are consecutive. 
Then each family corresponds to an identity sub-matrix $I_\nMembersSymmetric$ in $\testmatrix_2$. 
Now the problem becomes how to arrange the identity sub-matrices. \\

\noindent{\em Desirable Property 2:} The identity matrices corresponding to different families either appear in the same set of $\nMembersSymmetric$ rows in $\testmatrix_2$ or they do not appear in any shared rows.\\
{\it Intuition:} otherwise, a family would share tests with more other families. Then the probability that this family shares tests with infected families becomes larger. This would increase the probability that two infected families share tests after removing all the non-infected family columns, which in turn would increase the FP probability. 

\subsection{Proof of \Cref{lem:blockRowProb}}
\label{app:explain-Pjoint}

\begin{proof}
The probabilities can be explained as follows: 
\begin{enumerate}[(i)]
\item For $\Pr_{\text{joint}}^{I}$ in \eqref{prob-overlap1}, the numerator gives the number of possibilities that each block row contains at most one infected family, which is obtained by randomly choosing $\nDefFamilies$ block rows (the summation) and then from each chosen block row choosing one family to be infected ($\nOnesPerRow_i$ possible choices for $i$-th block row). 
The denominator is the total number of infection possibilities, and then the fraction denotes the probability that each block row contains at most one infected family. Thus, $\Pr_{\text{joint}}^{I}$ is obtained as the probability that there is some block row that contains two or more infected families. 
    
\item For $\Pr_{\text{joint}}^{II}$ in \eqref{prob-overlap2}, $(1-\familyInfectRate)^{\nOnesPerRow_i}$ is the probability that there is no infected family in the $i$-th block row, and $\nOnesPerRow_i\familyInfectRate(1-\familyInfectRate)^{\nOnesPerRow_i-1}$ is the probability that there is only one infected family in the $i$-th block row. The multiplication $\prod$ denotes the probability that any one block row contains at most one infected family. Thus, $\Pr_{\text{joint}}^{II}$ is obtained as the probability that there is some block row that contains two or more infected families. 
\end{enumerate}
\end{proof}

\subsection{Proof of \Cref{lemma_symmetric_c}}
\label{lemma_symmetric_c_proof}

\begin{proof}
Consider $\nOnesPerRow_i>\nOnesPerRow_j+1$, let $\nOnesPerRow_i'=\nOnesPerRow_i-1$ and $\nOnesPerRow_j'=\nOnesPerRow_j+1$. For the combinatorial model, we can verify the difference of the probability 
for $\nOnesPerRow_i'$ and $\nOnesPerRow_i$ by
\begin{align*}
\sum\limits_{\substack{|\mathcal{B}|=\nDefFamilies: \\ \mathcal{B}\subseteq\{1,2,\cdots,\nBlocks\}}} \prod\limits_{\ell\in\mathcal{B}}\nOnesPerRow_{\ell}'- \sum\limits_{\substack{|\mathcal{B}|=\nDefFamilies: \\ \mathcal{B}\subseteq\{1,2,\cdots,\nBlocks\}}} \prod\limits_{\ell\in\mathcal{B}}\nOnesPerRow_{\ell} &= (\nOnesPerRow_i'\nOnesPerRow_j'-\nOnesPerRow_i\nOnesPerRow_j)\cdot X \\
&=(\nOnesPerRow_i-\nOnesPerRow_j-1)\cdot X  \\
&>0,
\end{align*}
where $X$ is a positive value independent of $\nOnesPerRow_i$ and $\nOnesPerRow_j$. 
This implies that the minimum of the probability $\Pr_{\text{joint}}^{I}$ in \eqref{prob-overlap2} achieves its minimum roughly at the symmetric case where all $\nOnesPerRow_i$'s are equal, 
i.e., $\nOnesPerRow_i=\nOnesPerRow$ for all $i\in\{1,2,\cdots,\nBlocks\}$. 

Similarly, for the probabilistic model, consider the probability in \eqref{prob-overlap2}, we can calculate that
\begin{align}
&\prod_{\ell=1}^\nBlocks\left[(1-\familyInfectRate)^{\nOnesPerRow_{\ell}'}+\nOnesPerRow_{\ell}'\familyInfectRate(1-\familyInfectRate)^{\nOnesPerRow_{\ell}'-1}\right]  \nonumber \\
&\quad - \prod_{\ell=1}^\nBlocks\left[(1-\familyInfectRate)^{\nOnesPerRow_{\ell}}+\nOnesPerRow_{\ell}\familyInfectRate(1-\familyInfectRate)^{\nOnesPerRow_{\ell}-1}\right]  \\ \nonumber 
&=[(\nOnesPerRow_i-\nOnesPerRow_j)-(1-\familyInfectRate)^2]\familyInfectRate^2(1-\familyInfectRate)^{\nOnesPerRow_i+\nOnesPerRow_j-2}\cdot Y \nonumber \\
&>0,
\end{align}
where $Y=\prod_{\ell\neq i,j}\left[(1-\familyInfectRate)^{\nOnesPerRow_{\ell}}+\nOnesPerRow_{\ell}\familyInfectRate(1-\familyInfectRate)^{\nOnesPerRow_{\ell}-1}\right]>0$ is independent of $\nOnesPerRow_i$ and $\nOnesPerRow_j$. 
This implies that the minimum of the probability in \eqref{prob-overlap2} achieves its minimum roughly at the symmetric case where all $\nOnesPerRow_i$'s are equal, 
i.e., $\nOnesPerRow_i=\nOnesPerRow$ for all $i\in\{1,2,\cdots,\nBlocks\}$. 
\end{proof}

\subsection{Proof of \Cref{lemma_FP-prob}}
\label{lemma_FP-prob_proof}
The lemma is obtained under the assumption that the number of families $\nFamilies$ is a multiple of $\nBlocks$ and $\nOnesPerRow$. If $\nFamilies$ cannot be factorized, the error probabilities in \Cref{lemma_FP-prob} can be viewed as an upper bound  for the corresponding error probabilities. This can be seen by simply adding $\nFamilies'$ auxiliary families so that $\nFamilies+\nFamilies'=\nBlocks\nOnesPerRow$.

\begin{proof}
In the symmetric case, i.e., $\nOnesPerRow_i=\nOnesPerRow$ for all $i\in\{1,2,\cdots,\nBlocks\}$,
the probabilities in \eqref{prob-overlap1} and 
\eqref{prob-overlap2} become 
\begin{align}
\Pr_{\text{joint}}^{I}&=1-\frac{{\nBlocks \choose \nDefFamilies}\nOnesPerRow^{\nDefFamilies}}{{\nFamilies \choose \nDefFamilies}},   \label{prob-overlap1-symmetric}  \\
\Pr_{\text{joint}}^{II}&=1-\left((1-\familyInfectRate)^{\nOnesPerRow-1}(1-\familyInfectRate+\nOnesPerRow\familyInfectRate)\right)^\nBlocks.  \label{prob-overlap2symmetric}
\end{align}
For the symmetric combinatorial model, the number of infected members in an infected family $\nDefMembers=\nDefMembersSymmetric$ for all infected families $j$. 
If two families appear in the same set of $\nMembersSymmetric$ tests, the probability that all infected members in one family share the same $\nDefMembersSymmetric$ tests as the other family is simply 
\begin{align}
    \Pr(\text{no FP}|\text{joint}) =\frac{1}{{\nMembersSymmetric \choose \nDefMembersSymmetric}}. 
\end{align}
Thus the probability that FPs happen is 
\begin{align}
Pe&=\Pr(\text{FP}|\text{joint})\cdot \Pr_{\text{joint}}^{I}  =\left[1-\frac{1}{{\nMembersSymmetric \choose \nDefMembersSymmetric}}\right]\left[1-\frac{{\nBlocks \choose \nDefFamilies}\nOnesPerRow^{\nDefFamilies}}{{\nFamilies \choose \nDefFamilies}}\right]. 
\label{FP-prob1}
\end{align}
For the symmetric probabilistic model, the infection probability in an infected family $\memberInfectRate=\memberInfectRateSymmetric$ for all infected families $j$. 
If two families appear in the same set of $\nMembersSymmetric$ tests, then there is no false positives only when the two families have the same number of infected members and the infected (non-infected) members in one family must appear in the same set of tests as infected (non-infected) members of the other family. 
The probability that two families both have $i$ infected members is $\left[\memberInfectRateSymmetric^i(1-\memberInfectRateSymmetric)^{\nMembersSymmetric-i}\right]^2$, and the probability that all infected members in one family share tests with only infected members in the other family is simply $\frac{1}{{\nMembersSymmetric \choose i}}$. 
Thus, the probability that there is no false positives is given as follows, 
\begin{equation}
\Pr(\text{no FP}|\text{joint})=\sum_{i=1}^\nMembersSymmetric\left[\memberInfectRateSymmetric^i(1-\memberInfectRateSymmetric)^{\nMembersSymmetric-i}\right]^2 \frac{1}{{\nMembersSymmetric \choose i}}.
\end{equation}
Thus the probability that a false positive happens can be obtained as 
\begin{align}
Pe&=\Pr(\text{FP}|\text{joint})\cdot \Pr_{\text{joint}}^{II}  \nonumber \\ 
&=\left[\sum_{i=1}^\nMembersSymmetric\left[\memberInfectRateSymmetric^i(1-\memberInfectRateSymmetric)^{\nMembersSymmetric-i}\right]^2 \frac{1}{{\nMembersSymmetric \choose i}}\right] \nonumber \\
&\qquad \cdot \left[1-\left((1-\familyInfectRate)^{\nOnesPerRow-1}(1-\familyInfectRate+\nOnesPerRow\familyInfectRate)\right)^\nBlocks\right]. 
\label{FP-prob2}
\end{align}
Replacing $\nBlocks$ by $\nTests_2/\nMembersSymmetric$ and $\nOnesPerRow$ by $\nFamilies\nMembersSymmetric/\nTests_2$ completes the result.
\end{proof}

\subsection{Proof of \Cref{lemma-ErrorRate} and Discussions} \label{app:proof-lemma-ErrorRate}

\begin{proof}
For the combinatorial model (I), it is hard to explicitly calculate the expected error rate. The upper bound in \eqref{ErrorRate-1} is obtained by assuming that if there exist errors (FPs), then all non-infected members in infected families are misidentified as infected in the decoding of $\testmatrix_2$. (Note that all non-infected members in non-infected families are correctly identified by decoding of $\testmatrix_1$.)

For the probabilistic model (II), the upper bound for the expected error rate in \eqref{ErrorRate-2} is obtained by 
\begin{align}
&R_{II}(\text{error})  =\frac{1}{\nItems}\cdot \nBlocks \cdot  \left[\sum_{j=2}^\nOnesPerRow {\nOnesPerRow \choose j}\familyInfectRate^j (1-\familyInfectRate)^{\nOnesPerRow-j} \right. \nonumber \\
&\qquad\quad \left.\cdot 
\left(\sum_{i=1}^j{j \choose i}\memberInfectRateSymmetric^i (1-\memberInfectRateSymmetric)^{j-i}(j-i)\right)\cdot \nMembersSymmetric\right]  \label{ErrorRate-2-pf-1}  \\
&=\frac{\nBlocks\nMembersSymmetric}{\nItems}\cdot \Bigg[\sum_{j=2}^\nOnesPerRow {\nOnesPerRow \choose j}\familyInfectRate^j (1-\familyInfectRate)^{\nOnesPerRow-j}  \nonumber \\
&\qquad\qquad \cdot \big(j(1-\memberInfectRateSymmetric)-j(1-\memberInfectRateSymmetric)^j\big)\Bigg]  \label{ErrorRate-2-pf-2} \\
&<\frac{(1-\memberInfectRateSymmetric)\nTests_2}{\nItems}\cdot \left[\sum_{j=2}^\nOnesPerRow {\nOnesPerRow \choose j}\familyInfectRate^j (1-\familyInfectRate)^{\nOnesPerRow-j} \cdot j\right]  \nonumber \\
&=\frac{(1-\memberInfectRateSymmetric)\nTests_2}{\nItems}\cdot \big[\nOnesPerRow\familyInfectRate-\nOnesPerRow\familyInfectRate(1-\familyInfectRate)^{\nOnesPerRow-1}\big], \nonumber \\
&=(1-\memberInfectRateSymmetric)\familyInfectRate\big[1-(1-\familyInfectRate)^{\nOnesPerRow-1}\big],  \label{ErrorRate-2-pf-3}
\end{align}
where the expression in the bracket in \eqref{ErrorRate-2-pf-1} for each~$j$ denotes the expected number of FPs in one block row if there are $j$ families infected in this block row, \eqref{ErrorRate-2-pf-2} is obtained from the expected value of binomial distribution, and \eqref{ErrorRate-2-pf-3} follows by substituting $\nOnesPerRow=\frac{\nItems}{\nTests_2}$. 
\end{proof}

We here make the following observation about the system FP probability $\Pr(\text{any-FP})$:
As we explore further in \Cref{section-experiments} non-adaptive group testing requires more tests than adaptive.
Assume that $\nDefFamilies = \Theta(\nFamilies^\sparseRegimeFamilyPar)$ for $\sparseRegimeFamilyPar \in [0,1)$ and choose $\nSamples = \nMembersSymmetric-1$ in \Cref{algorithm-noiseless}. 
Adaptive testing allows to achieve zero error with $\nDefFamilies\log_2\nFamilies+\nDefFamilies\nMembersSymmetric$ tests;
if we use the same (order) number of tests with a non-adaptive strategy, i.e., $T_1=\nDefFamilies\log_2\frac{\nFamilies}{\nDefFamilies}$ and $T_2=\nDefFamilies\left(\log_2\nDefFamilies+\nMembersSymmetric\right)$, we get $\Pr(\text{any-FP})$ in \Cref{lemma_FP-prob} 
approximately equal to $\left(1-\frac{1}{\nMembersSymmetric}\right)\left[1-\frac{{\nTests_2/\nMembersSymmetric \choose \nDefFamilies}}{\left(\frac{\nTests_2/\nMembersSymmetric}{\nDefFamilies}\right)^\nDefFamilies}\frac{\left(\nFamilies/\nDefFamilies\right)^{\nDefFamilies}}{{\nFamilies \choose \nDefFamilies}}\right]$ which is bounded away from 0. 
The latter can be seen as follows: i) $\nTests_2/\nMembersSymmetric\approx \nDefFamilies\ll\nFamilies$; ii) $\frac{{n \choose k}}{\left(\frac{n}{k}\right)^ k} \big/ \frac{{n+m \choose k}}{\left(\frac{n+m}{k}\right)^ k}=\left(\frac{n}{n+m}\right)^k\cdot\prod_{i=1}^m\frac{n+i-k}{n+i}$ is decreasing with $m$ and can be very small when $m\gg n$.

Fig.~\ref{fig_FPprob} depicts $\Pr(\text{any-FP})$ and $R(\text{error})$ for parameters $\nFamilies=64$, $\nDefFamilies=6$, $\nDefMembersSymmetric=4$, $\nMembersSymmetric=5$, $\familyInfectRate= 1/8$, and $\memberInfectRateSymmetric=0.8$.
\begin{figure}[!t]
\centering
\includegraphics[width=0.8\linewidth, height=3.5cm]{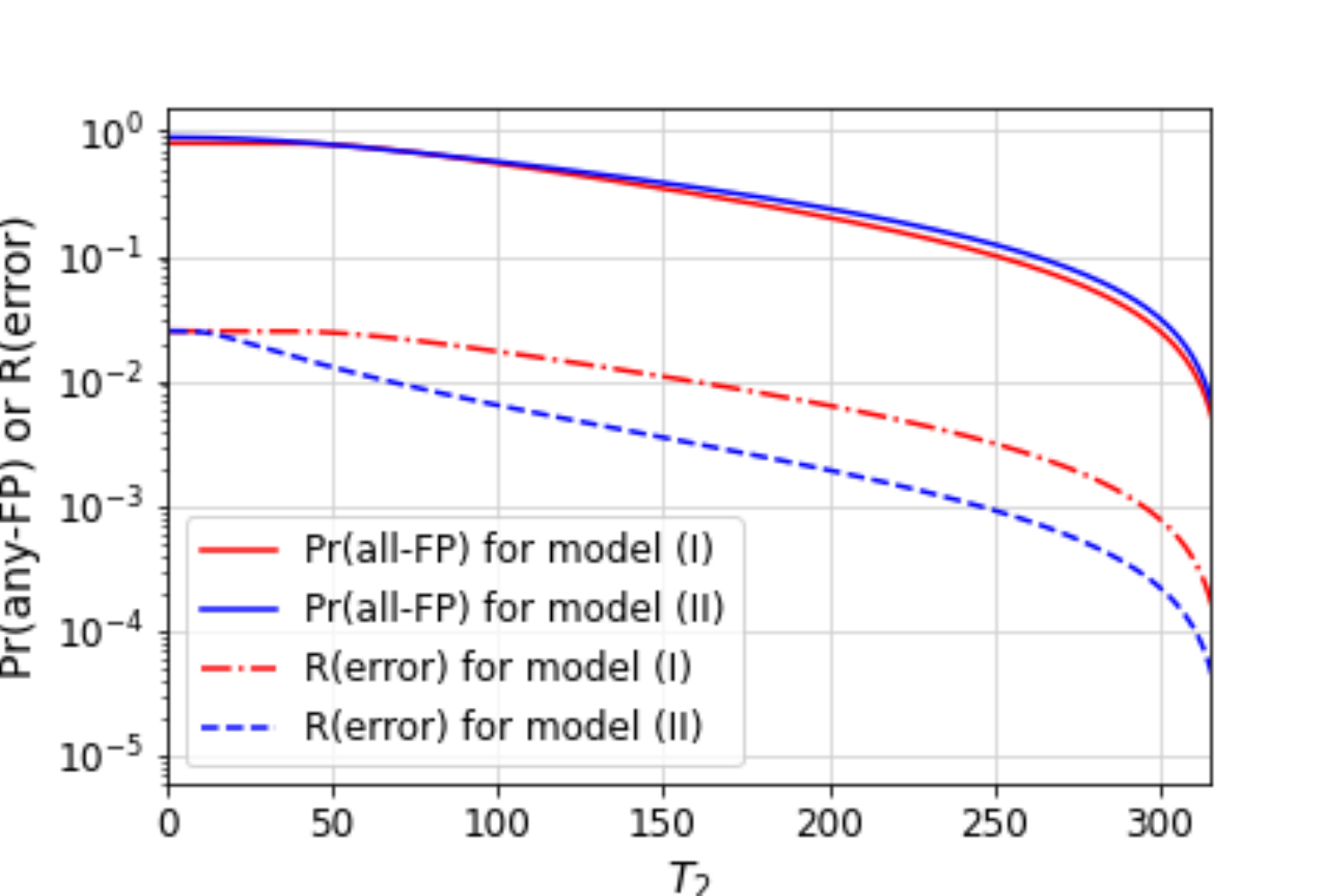}
\caption{System FP probability and FP error rate.}
\label{fig_FPprob}
\end{figure}

%

\section{Appendix for \Cref{sec:LBP}: Loopy Belief Propagation algorithm}
\label{app:sec:LBP}
We here describe our loopy belief propagation algorithm (LBP) and update rules for our probabilistic model (II). We use the factor graph framework of \cite{kschischang2001factor} and derive closed-form expressions for the sum-product update rules (see equations (5) and (6) in \cite{kschischang2001factor}).

The LBP algorithm on a factor graph iteratively exchanges messages across the variable and factor nodes. The messages to and from a variable node $\defFamilyVariable_\familyIndex$ or $\defVariable_{\memberIndex}$ are \textit{beliefs} about the variable or distributions (a local estimate of $\Pr(\defFamilyVariable_\familyIndex|\text{observations})$ or $\Pr(\defVariable_{\memberIndex}|\text{observations})$). Since all the random variables are binary, in our case each message would be a 2-dimensional vector $[a,b]$ where $a,b \geq 0$. 
Suppose the result of each test is $y_\testIndex$, i.e., $\testResult_{\testIndex}=y_{\testIndex}$ and we wish to compute the marginals $\Pr(\defFamilyVariable_\familyIndex=\defectFamilyValue|\testResult_{1}=y_1,\testResult_{2}=y_2,...,\testResult_{\nTests}=y_T)$  and  $\Pr(\defVariable_{\memberIndex}=\defectValue|\testResult_{1}=y_1,\testResult_{2}=y_2,...,\testResult_{\nTests}=y_\nTests)$ for $\defectFamilyValue,\defectValue\in\{0,1\}$. 
The LBP algorithm proceeds as follows:

\begin{figure*}[h!]
	\centering
	\captionsetup{justification=centering}
	\includegraphics[height=12cm, width=0.8\linewidth]{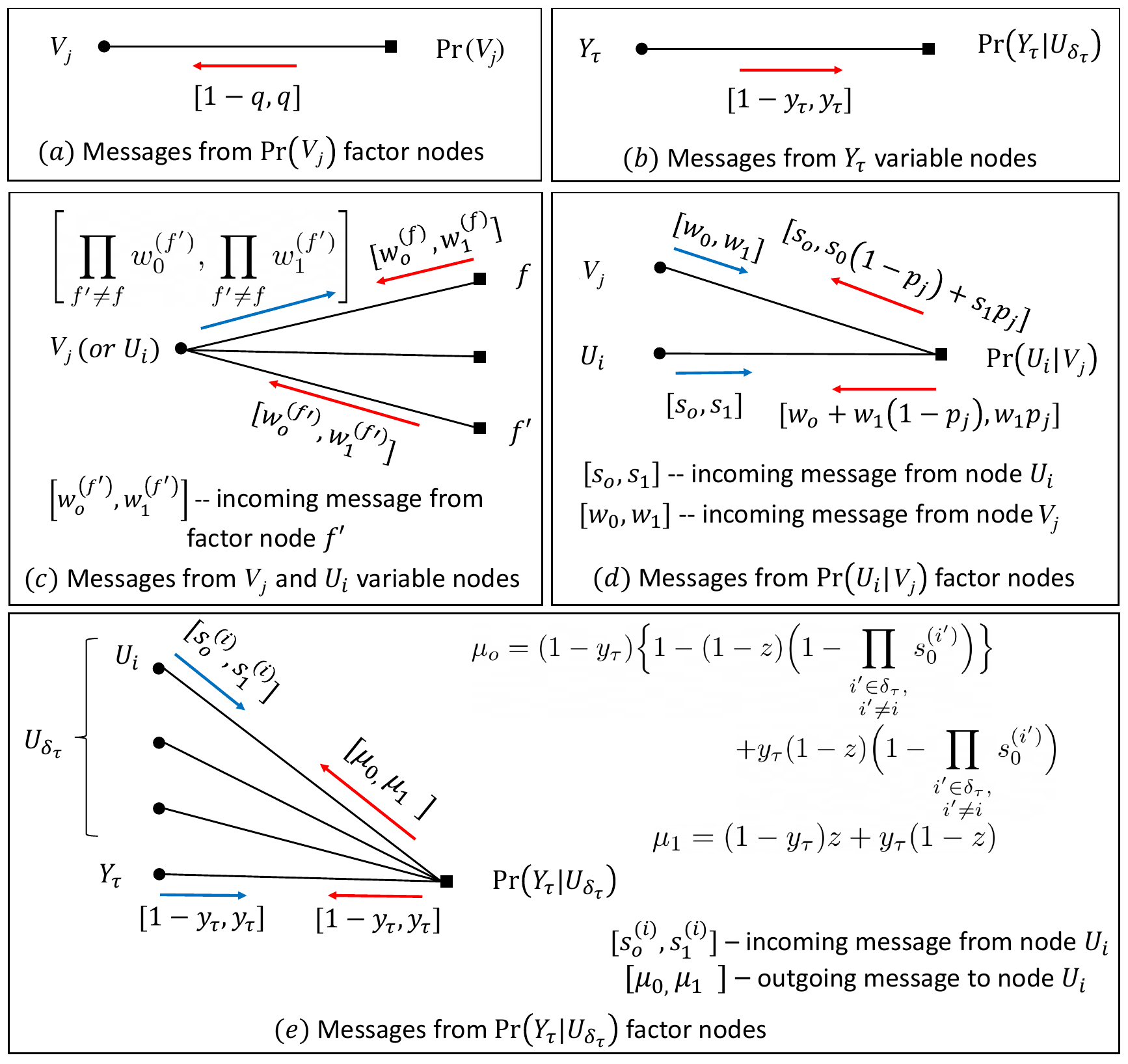}
	\caption{The update rules for the factor and variable node messages.}
	\vspace{-0.2cm}
	\label{fig:messages-appendix}
\end{figure*}

\begin{enumerate}
	\item \textit{Initialization:} The variable nodes $\defFamilyVariable_\familyIndex$ and $\defVariable_{\memberIndex}$ transmit the message $[0.5,0.5]$ on each of their incident edges. Each variable node $Y_\tau$ transmits the message $[1-y_{\tau},y_{\tau}]$, where $y_{\tau}$ is the observed test result, on its incident edge.
	
	\item \textit{Factor node messages:} Each factor node receives the messages from the neighboring variable nodes and computes a new set of messages to send on each incident edge. The rules on how to compute these messages are described next.
	
	\item \textit{Iteration and completion.} The algorithm alternates between steps 2 and 3 above a fixed number of times (in practice 10 or 20 times works well) and computes an estimate of the posterior marginals as follows -- for each variable node $\defFamilyVariable_\familyIndex$ and $\defVariable_{\memberIndex}$, we take the coordinatewise product of the incoming factor messages and normalize to obtain an estimate of $\Pr(\defFamilyVariable_\familyIndex=\defectFamilyValue|y_1...y_\nTests)$ and $\Pr(\defVariable_{\memberIndex}=\defectValue|y_1...y_\nTests)$ for $\defectFamilyValue,\defectValue\in\{0,1\}$.
\end{enumerate}

Next we describe the simplified variable and factor node message update rules. We use equations (5) and (6) of \cite{kschischang2001factor} to compute the messages. 

\textit{Leaf node messages:} At every iteration, the variable node $\testResult_\testIndex$ continually transmits the message $[0,1]$ if $\testResult_\testIndex=1$ and $[1,0]$ if $\testResult_\testIndex=0$ on its incident edge. The factor node $\Pr(\defFamilyVariable_\familyIndex)$ continually transmits $[1-\familyInfectRate,\familyInfectRate]$ on its incident edge; see Fig.~\ref{fig:messages-appendix} (a) and (b).

\textit{Variable node messages:} The other variable nodes $\defFamilyVariable\familyIndex$ and $\defVariable_{\memberIndex}$ use the following rule to transmit messages along the incident edges: for incident each edge $e$, a variable node takes the elementwise product of the messages from every other incident edge $e'$ and transmits this along $e$; see Fig.~\ref{fig:messages-appendix} (c).

\textit{Factor node messages:} For the factor node messages, we calculate closed form expressions for the sum-product update rule (equation (6) in \cite{kschischang2001factor}). The simplified expressions are summarized in Fig.~\ref{fig:messages-appendix} (d) and (e). Next we briefly describe these calculations. 

Firstly, we note that each message represents a probability distribution. One could, without loss of generality, normalize each message before transmission. Therefore, we assume that each message $\mu=[a,b]$ is such that $a+b=1$. Now, the the leaf nodes labeled $\Pr(V_j)$ perennially transmit the prior distribution corresponding to
$V_j$. 

Next, consider the factor node $\Pr(U_i|V_j)$ as shown in Fig.~\ref{fig:messages-appendix} (d). The message sent to $U_i$ is calculated as 
\begin{align*}
\mu_u &=\sum_{v\in \{0,1\}} \Pr(U_i=u|V_j=v) w_v\\
&= w_0 (1-u) +w_1 p_j^u (1-p_j)^{1-u}. 
\end{align*}
Similarly, the message sent to $V_i$ is 
\begin{align*}
\nu_v &=\sum_{u\in \{0,1\}} \Pr(U_i=u|V_j=v) s_u\\
&= s_0 (v(1-p_j)+1-v) + s_1 v p_j. 
\end{align*}

Finally for the factor nodes $\Pr(Y_{\tau}|U_{\delta_{\tau}})$ as shown in Fig.~\ref{fig:messages-appendix} (e), note that the messages to $Y_{\tau}$ play no role since they are never used to recompute the variable messages. The messages to $U_i$ nodes are expressed as 
\begin{align*}
\mu_u &= \sum_{\substack{y\in\{0,1\},\\ \{u_{i'}\in \{0,1\}: i'\in \delta_{\tau}\setminus\{i\}\}}} \Big( \Pr(Y_{\tau}=y|U_{\delta_{\tau}}=u_{\delta_{\tau}})\\ &\hspace{3cm}(1-y_{\tau})^{1-y} y_{\tau}^y \prod_{i'\in \delta_{\tau}\setminus\{i\}\}} s_{u_{i'}}^{(i')} \Big)\\
 &= (1-y_{\tau})\sum_{\substack{\{u_{i'}\in \{0,1\}:\\ i'\in \delta_{\tau}\setminus\{i\}\}}} \Big( \Pr(Y_{\tau}=0|U_{\delta_{\tau}}=u_{\delta_{\tau}}) \\
 &\hspace{4cm} \prod_{i'\in \delta_{\tau}\setminus\{i\}\}} s_{u_{i'}}^{(i')} \Big)\\
  &+ y_{\tau}\sum_{\substack{\{u_{i'}\in \{0,1\}:\\ i'\in \delta_{\tau}\setminus\{i\}\}}} \Big( \Pr(Y_{\tau}=1|U_{\delta_{\tau}}=u_{\delta_{\tau}}) \prod_{i'\in \delta_{\tau}\setminus\{i\}\}} s_{u_{i'}}^{(i')} \Big).
\end{align*}
From our Z-channel model, recall that $\Pr(Y_{\tau}=0|U_{\delta_{\tau}}=u_{\delta_{\tau}}) = 1$ if $u_i=0\ \forall\ i \in \delta_{\tau}$ and $\Pr(Y_{\tau}=0|U_{\delta_{\tau}}=u_{\delta_{\tau}}) = z$ otherwise. Thus we split the summation terms into 2 cases -- one where  $u_{i'}=0$ for all $i'$ and the other its complement. Also combining this with the assumption that the messages are normalized, i.e., $s^{(i)}_0+s^{(i)}_1=1,$ we get 
\begin{align*}
    \sum_{\substack{\{u_{i'}\in \{0,1\}:\\ i'\in \delta_{\tau}\setminus\{i\}\}}} &\Big( \Pr(Y_{\tau}=0|U_{\delta_{\tau}}=u_{\delta_{\tau}}) \prod_{i'\in \delta_{\tau}\setminus\{i\}\}} s_{u_{i'}}^{(i')} \Big)\\
    &= \mathbbm{1}_{u=1}z + \mathbbm{1}_{u=0}\Big\{ 1-(1-z)(1-\prod_{\substack{i'\in\delta_{\tau}\\i'\neq i}} s^{(i')}_0) \Big\},
\end{align*}
and
\begin{align*}
    \sum_{\substack{\{u_{i'}\in \{0,1\}:\\ i'\in \delta_{\tau}\setminus\{i\}\}}} &\Big( \Pr(Y_{\tau}=1|U_{\delta_{\tau}}=u_{\delta_{\tau}}) \prod_{i'\in \delta_{\tau}\setminus\{i\}\}} s_{u_{i'}}^{(i')} \Big)\\
    &= \mathbbm{1}_{u=1}(1-z) + \mathbbm{1}_{u=0}\Big( (1-z)(1-\prod_{\substack{i'\in\delta_{\tau}\\i'\neq i}} s^{(i')}_0) \Big).
\end{align*}
Substituting $u=0$, and $u=1$ we obtain the messages
\begin{align*}
\mu_0 &= (1-y_{\tau})\Big\{ 1-(1-z)(1-\prod_{\substack{i'\in\delta_{\tau}\\i'\neq i}} s^{(i')}_0) \Big\}\\
&+ y_{\tau} (1-z)(1-\prod_{\substack{i'\in\delta_{\tau}\\i'\neq i}} s^{(i')}_0),
\end{align*}
and
\begin{align*}
\mu_1 &= (1-y_{\tau})z+y_{\tau}(1-z).
\end{align*}
For our probabilistic model, the complexity of computing the factor node messages increases only linearly with the factor node degree. 

\section{Appendix for \Cref{section-experiments}: Other Results}
\label{app:sec:eval}
We next provide additional experimental results to the ones provided in Section~\ref{section-experiments}.

(i)~\textit{Noiseless testing -- Average number of tests:}
In \Cref{fig:exp1}, we reproduce additional numerics akin to the ones in Section~\ref{section-experiments} for number of tests in the noiseless-testing case.
As earlier, we measure the average number of tests needed by 3 algorithms that achieve zero-error reconstruction
(Alg.~\ref{algorithm-noiseless} with $\nSamples = 1$, Alg.~\ref{algorithm-noiseless} with $\nSamples = \nMembersSymmetric$,
and
classic BSA),
and a version of our nonadaptive algorithm (\Cref{sec:nonadapt}) that uses $\nTests_1 = \nFamilies$ tests for submatrix $\testmatrix_1$ and has an overall FP rate around $0.5\%$.
Alg.~\ref{algorithm-noiseless} assumes no prior knowledge of the number of infected families/classes or members/students, 
hence uses BSA as group-testing algorithm for the $\adapt()$ function.

\begin{figure}
	\centering
	\captionsetup{justification=centering}	
\begin{subfigure}{0.38\textwidth}
	\centering
	\captionsetup{justification=centering}	
	\includegraphics[ height=3.8cm, width = \textwidth]{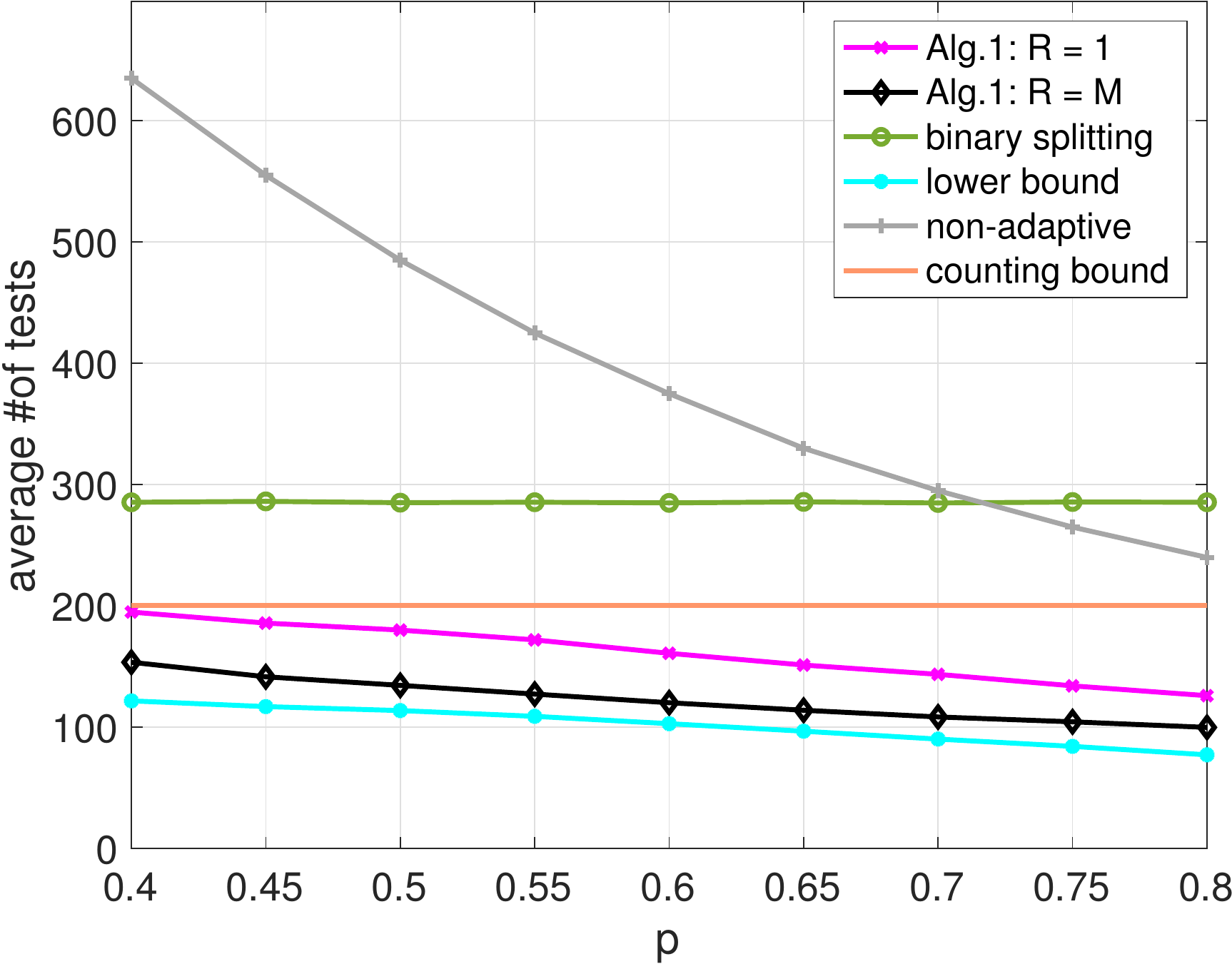}
	\caption{Community 1---sparse regime.}
	\label{fig:nTests_Com1_sparse}
	\vspace{0.3cm}
\end{subfigure}

\begin{subfigure}{0.38\textwidth}
	\centering
	\captionsetup{justification=centering}	
	\includegraphics[ height=3.8cm, width = \textwidth]{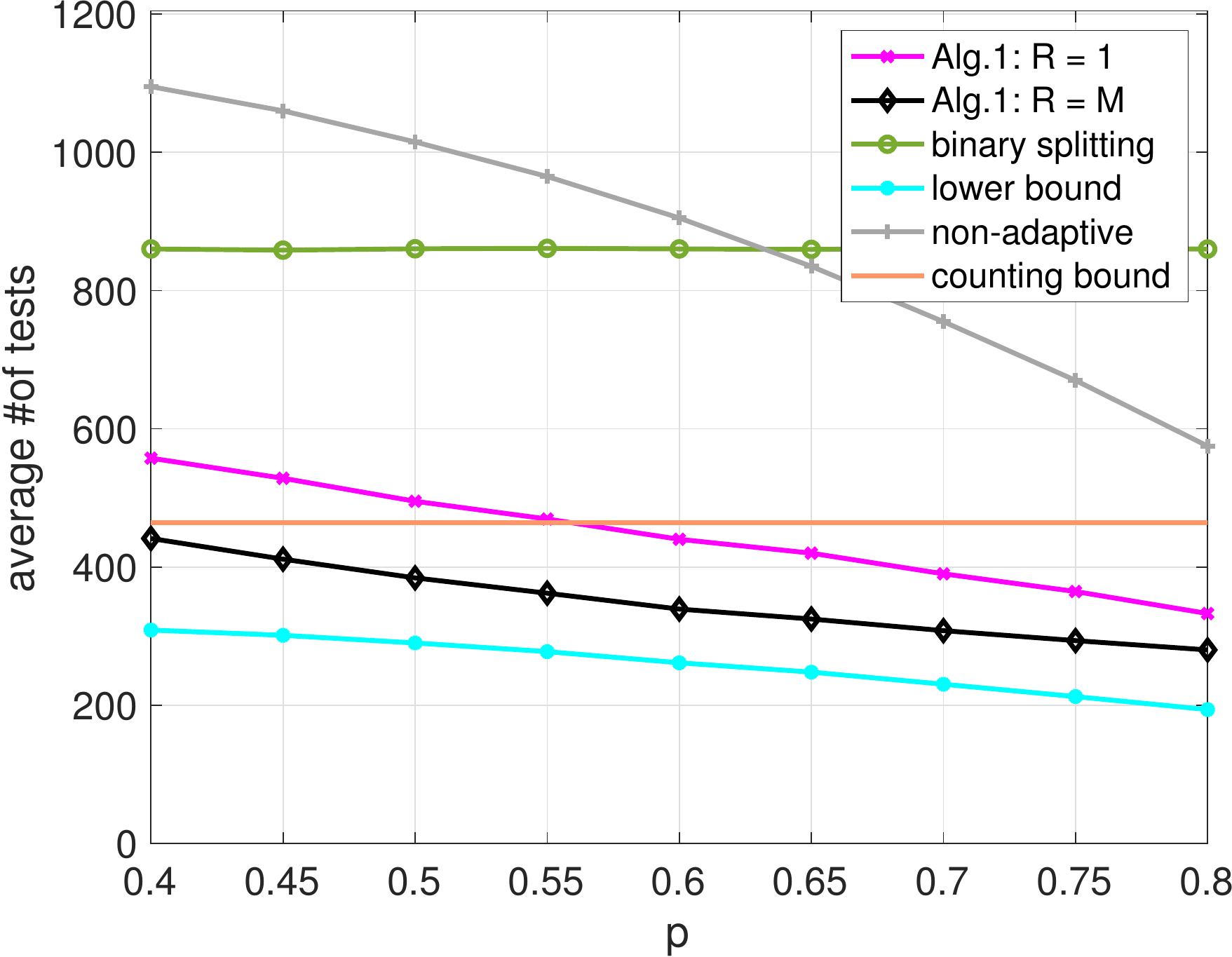}
	\caption{Community 1---linear regime.}
	\label{fig:nTests_Com1_linear}
	\vspace{0.3cm}
\end{subfigure}

\begin{subfigure}{0.38\textwidth}
	\centering
	\captionsetup{justification=centering}	
	\includegraphics[ height=3.8cm, width = \textwidth]{figs/pVarying-sparse-20-50-eps-converted-to.pdf}
	\caption{Community 2---sparse regime.}
	\label{fig:nTests_Com2_sparse}
	\vspace{0.3cm}
\end{subfigure}

\begin{subfigure}{0.38\textwidth}
	\centering
	\captionsetup{justification=centering}	
	\includegraphics[ height=3.8cm, width = \textwidth]{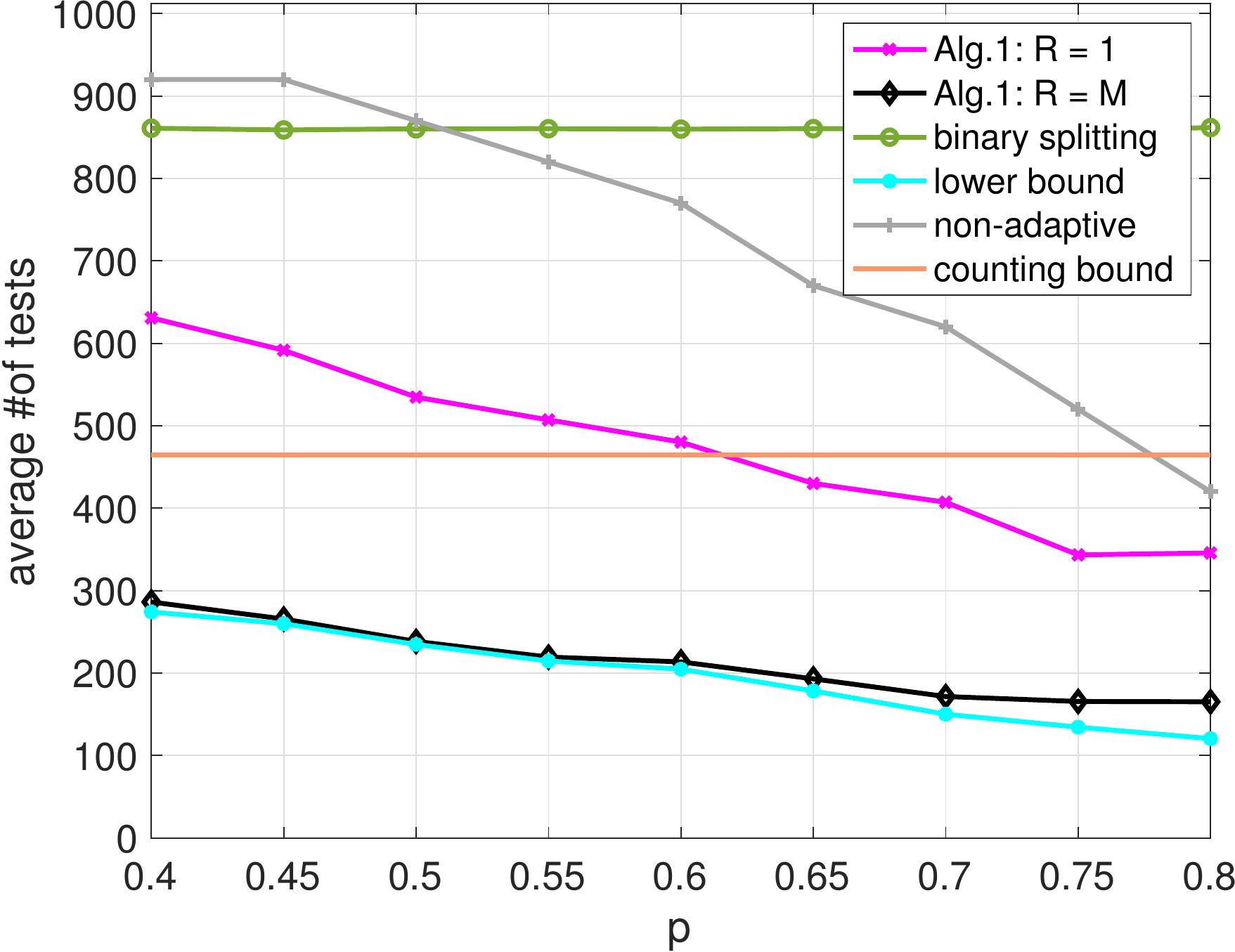}
	\caption{Community 2---linear regime.}
	\label{fig:nTests_Com2_linear}
	\vspace{0.3cm}
\end{subfigure}
\caption{Experiment (i): \\Noiseless case---Average number of tests.}
\label{fig:exp1}
\end{figure}

Fig.~\ref{fig:exp1} depicts our results:
We observe that both versions of Alg.~\ref{algorithm-noiseless} (black and magenta lines) need significantly fewer tests compared to classic BSA (green line),
while staying below the counting bound.
This indicates the potential benefits from the community structure,
even when the number of infected members is unknown.
More interestingly, when $\nSamples = \nMembersSymmetric$, Alg.~\ref{algorithm-noiseless}
performs close to the lower bound in most realistic scenarios $\memberInfectRateSymmetric \in [0.5,0.8]$ (as also shown in \Cref{sec:adapt}). 
The grey line shows number of tests needed by our
nonadaptive algorithm; we observe that even that algorithm needs fewer tests than BSA when $\memberInfectRateSymmetric$ gets larger than $0.5$, of course at the cost of a (FP) error rate of $0.5\%$.

(ii)~\textit{Noiseless testing -- Average error rate:}
In Fig.~\ref{fig:exp2}, we reproduce additional numerics akin to the ones in Section~\ref{section-experiments} for average error rates in the noiseless-testing case.
As earlier, we quantify the additional cost in terms of error rate, 
when one goes from a two-stage adaptive algorithm that achieves zero-error identification to much faster single-stage nonadaptive algorithms.

Fig.~\ref{fig:Noiseless_p08} is a reproduction of Fig.~\ref{fig:noiseless-T} for $p = 0.8$, and as can be seen its behavior is very similar to  Fig.~\ref{fig:noiseless-T}.

Fig.~\ref{fig:Noiseless_linear_2stage_test} depicts the FP and FN error rates (averaged over $500$ runs)
as a function of $\memberInfectRateSymmetric \in [0.3,0.8]$ for Community~1 for the linear regime.
We observe that any community-aware nonadaptive algorithm performs better than traditional nonadaptive group testing (red line) when $\memberInfectRateSymmetric > 0.5$ -- the absolute performance gap ranges from $0.2\%$ (when $\memberInfectRateSymmetric = 0.5$) to $8.5\%$ (when $\memberInfectRateSymmetric = 0.8$).
``COMP with C-encoder'' has a stable FP rate across for all $\memberInfectRateSymmetric$ values that was close to $2\%$, and a zero FN rate by construction.
Unlike the sparse regime, the LBP consistently produces better error rates compared to the COMP decoder. However, for low values of $p$, LBP produces more FN errors. For $\memberInfectRateSymmetric>0.6$, both the FN and FP error rates are close to 0 for LBP.

Fig.~\ref{fig:Noiseless_linear_p06} and Fig.~\ref{fig:Noiseless_linear_p08} examine the effect of the number of tests in the linear regime.
For $p=0.6$, ``C-LBP with NC-encoder'' performs better than ``COMP with C-encoder'' for $T>450$ until which both have high error rates. On the other hand, for $p=0.8$, ``C-LBP with NC-encoder'' performs better than ``COMP with C-encoder'' for all values of $T$.
More importantly, ``COMP with C-encoder'' seems to saturate to  a non-zero FP error rate, while  ``C-LBP with NC-encoder'' is able to attain close to zero error FP and FN rates. These results contrast with the results for the sparse regime.

(iii)~\textit{Noisy testing:}

In \Cref{fig:exp3}, we reproduce additional numerics akin to the ones in Section~\ref{section-experiments} for average error rates in the noisy-testing case.
As earlier, we assuming the Z-channel noise of \Cref{subsec:Noise-Error} with parameter $\znoiseProb = 0.15$, 
and we evaluate the performance of our community-based LBP decoder of \Cref{sec:LBP} against a LBP that does not account for community---namely its factor graph has no $\defFamilyVariable_\familyIndex$ nodes.

Fig.~\ref{fig:Noisy_p06} is a reproduction of Fig.~\ref{fig:noisy} for  $p = 0.6$, and as can be seen its behavior is very similar to it.

Fig.~\ref{fig:Noisy_linear_p06} and Fig.~\ref{fig:Noisy_linear_p08} depict our results for Community 1 and
for $\memberInfectRateSymmetric = 0.6$ and $\memberInfectRateSymmetric = 0.8$ in the linear infection regime.
We observe that the knowledge of the community structure reduces the FN rates achieved by LBP. The FP error rates are always close to 0 while the, FN error rates drop significantly (up to $60\%$ when tests are few), which is important in our context since FN errors lead to further infections.

(iv) \textit{Asymmetric case---Linear regime:}

\pavlos{Here we offer the results about an asymmetric setup that parallels the one of \Cref{section-experiments}. 
	Infections follow again the probabilistic model (II), and the size of each family is randomly selected from the interval $[5, 50]$ and the infection rate of each infected family is randomly selected from the range $[0.4, 0.8]$. 
	But, this time $\familyInfectRate = 5\%$.
	
	\Cref{fig:assymteric} depicts our results.
	BSA needs on average $6.19\times$ (that can
	reach up to $13.87\times$) more tests compared to the probabilistic bound, while the two versions of \Cref{algorithm-noiseless} with $\nSamples = 1$ and $\nSamples = \nMembersSymmetric$ need only $2.74\times$ and $1.19\times$ (that can reach up to
	$9.7\times$ and $2.03\times$) more tests, respectively.
	Also, similarly to the sparse regime, there is a significantly smaller range between the $25$-th and $75$-th percentiles of the boxplots related to \Cref{algorithm-noiseless} that indicates its more predictable performance compared to BSA. 
}

\begin{figure}
	\centering
	\captionsetup{justification=centering}	
	\begin{subfigure}{0.38\textwidth}
		\centering
		\captionsetup{justification=centering}	
		\includegraphics[height=3.8cm, width = \textwidth]{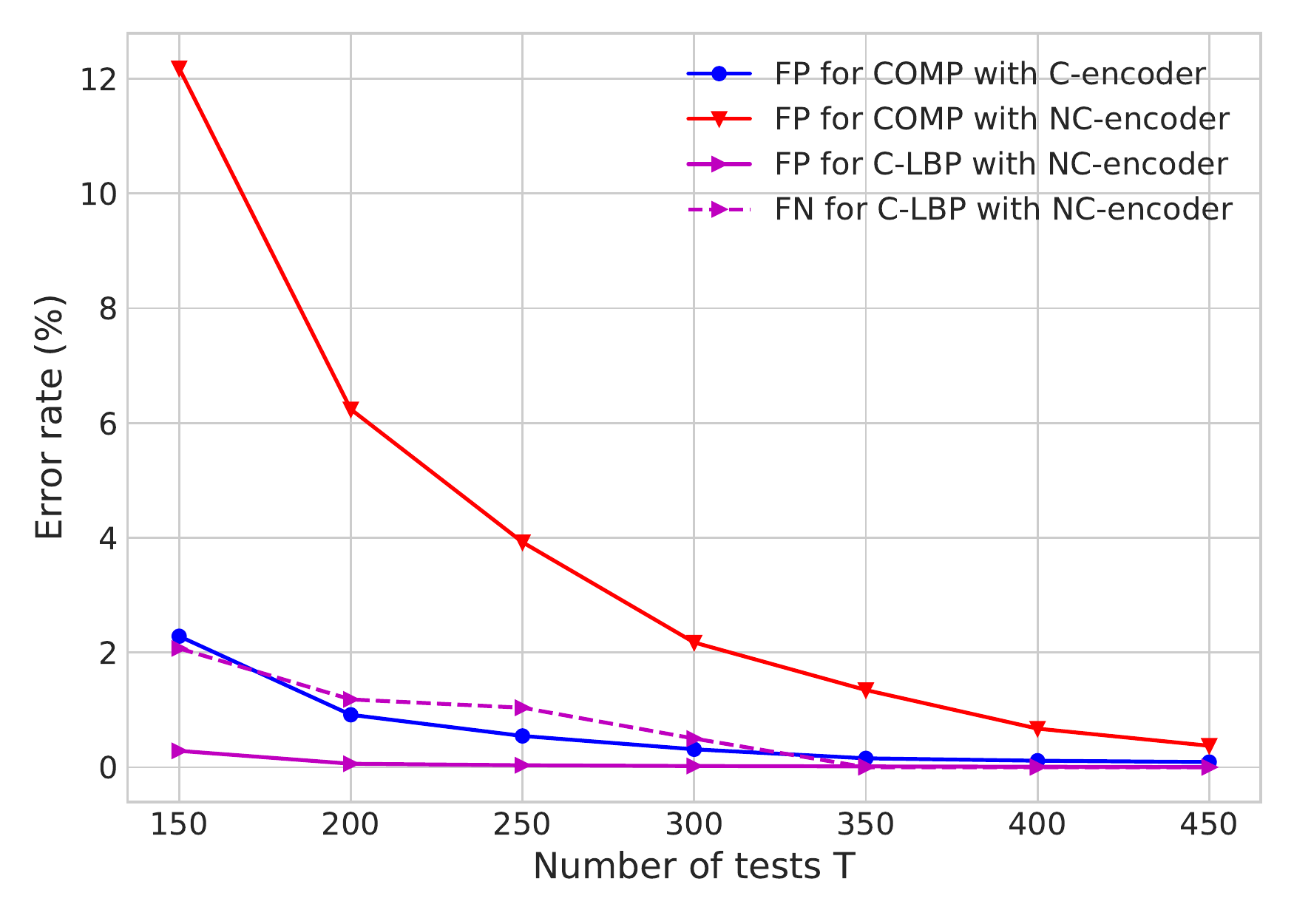}
		\caption{Noiseless case: Average error rate $p=0.8$ for sparse regime.}
		\label{fig:Noiseless_p08}
	\end{subfigure}
	
	\begin{subfigure}{0.38\textwidth}
		\centering
		\captionsetup{justification=centering}			\includegraphics[height=3.8cm, width = \textwidth]{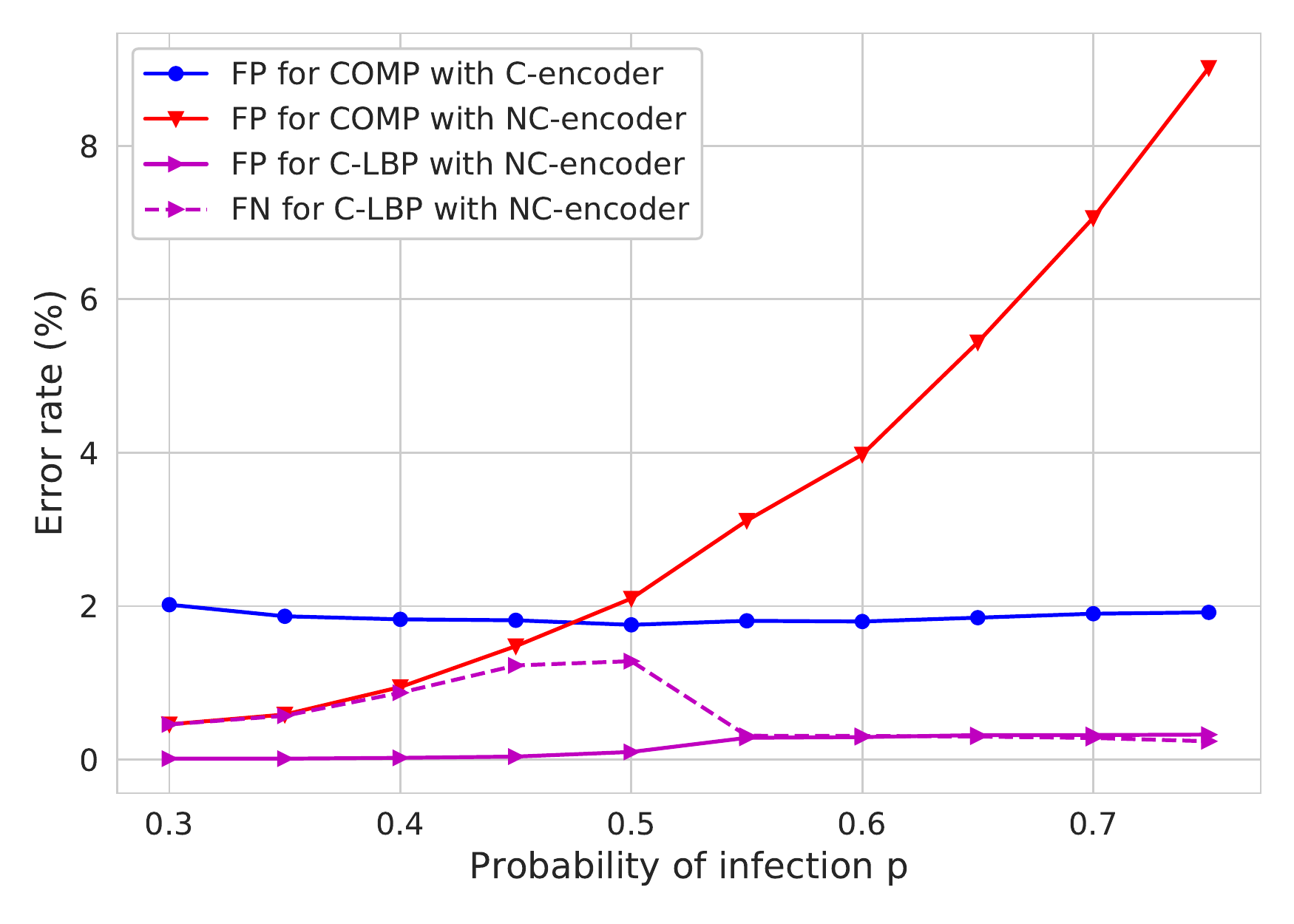}
		\caption{Noiseless case: Average error rate with few tests for linear regime.}
		\label{fig:Noiseless_linear_2stage_test}
	\end{subfigure}
	
	\begin{subfigure}{0.38\textwidth}
		\centering
		\captionsetup{justification=centering}			\includegraphics[height=3.8cm, width = \textwidth]{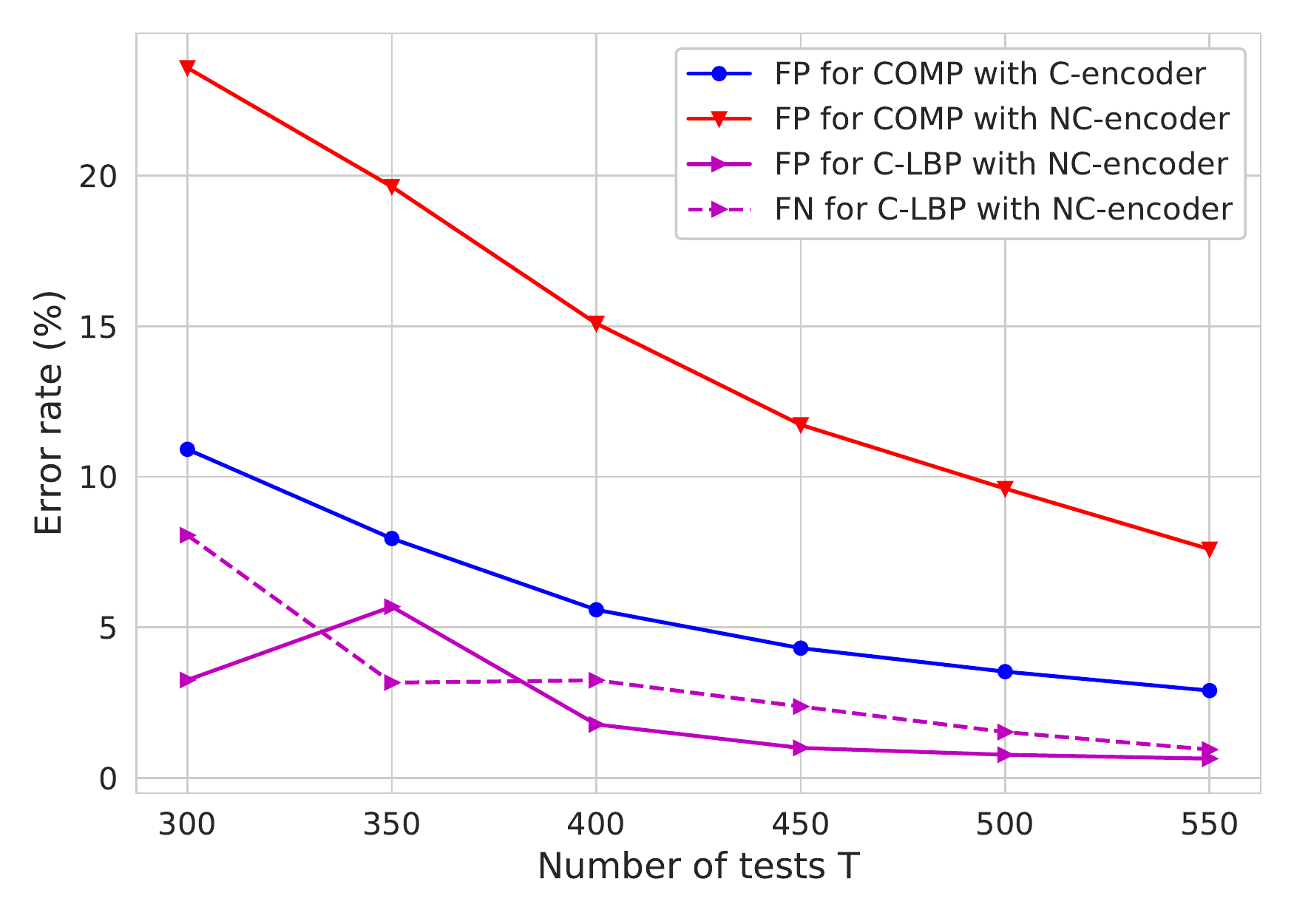}
		\caption{Noiseless case: Average error rate $p=0.6$ for linear regime.}
		\label{fig:Noiseless_linear_p06}
	\end{subfigure}
	
	\begin{subfigure}{0.38\textwidth}
		\centering
		\captionsetup{justification=centering}			\includegraphics[height=3.8cm, width = \textwidth]{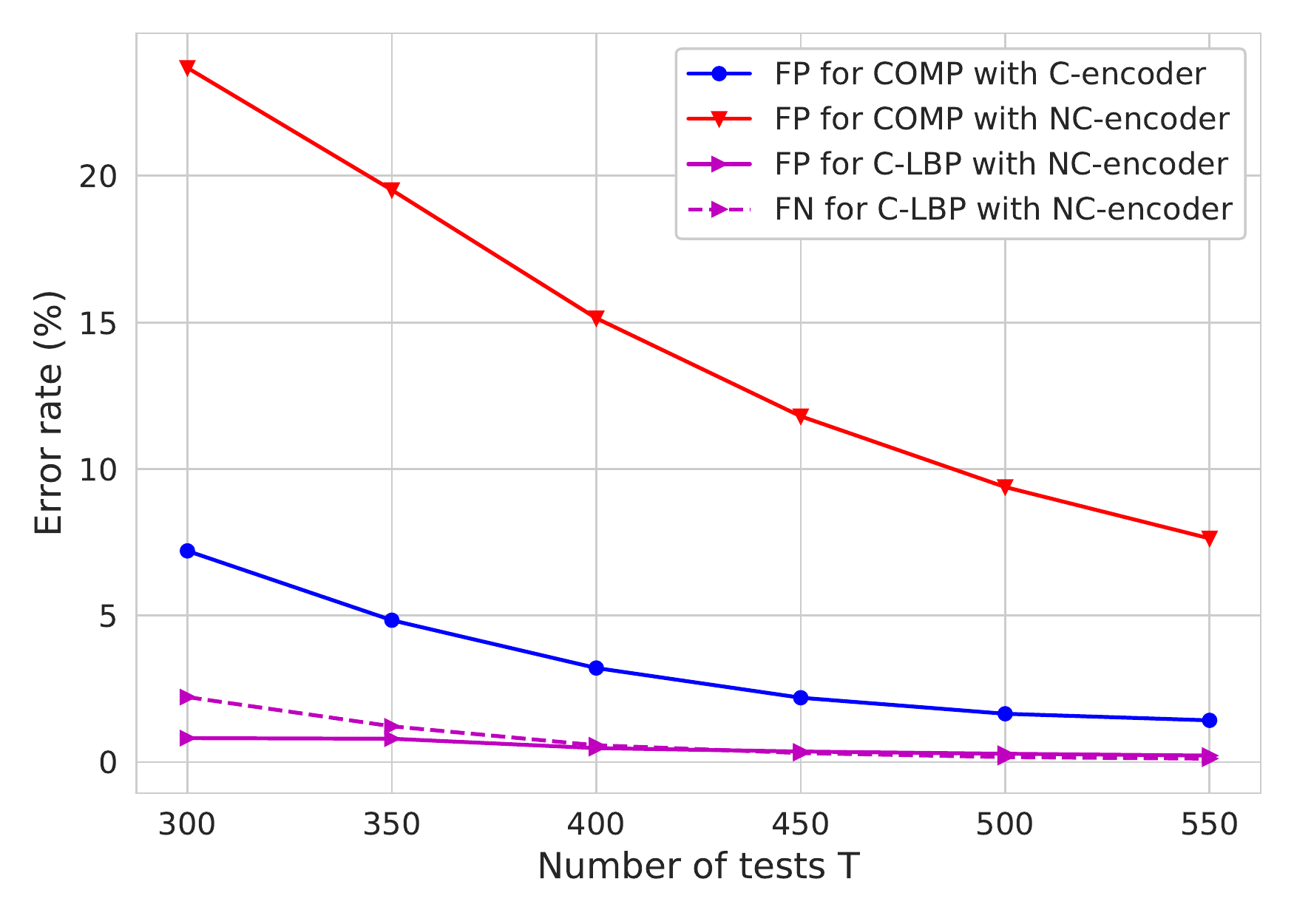}
		\caption{Noiseless case: Average error rate $p=0.8$ for linear regime.}
		\label{fig:Noiseless_linear_p08}
	\end{subfigure}
	\caption{Experiment (ii): \\Noiseless case---Average error rate.}
	\label{fig:exp2}
\end{figure}

\begin{figure}
	\centering
	\captionsetup{justification=centering}
	\begin{subfigure}{0.38\textwidth}
		\centering
		\captionsetup{justification=centering}		\includegraphics[height=3.5cm, width = \textwidth]{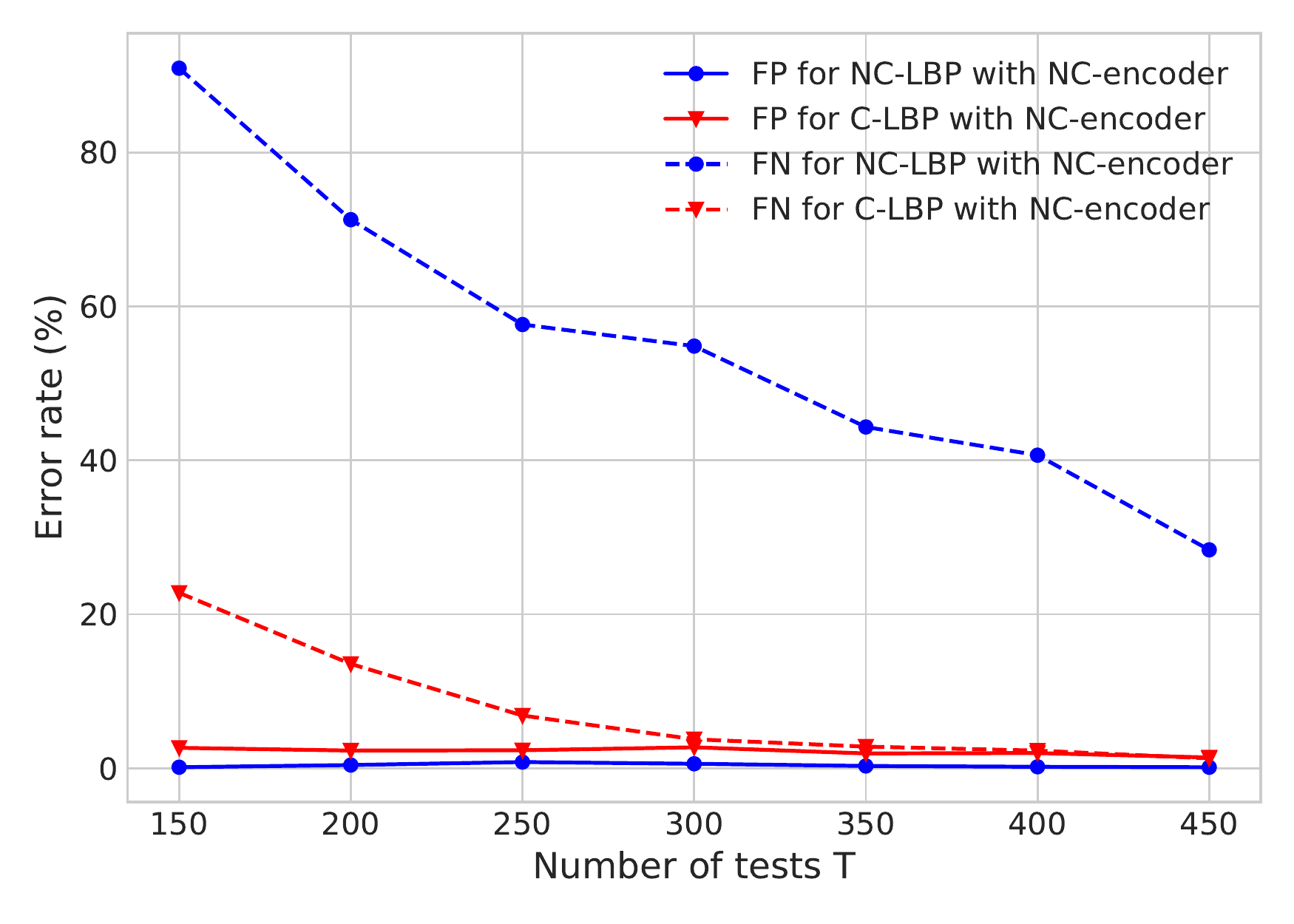}
		\caption{Noisy case: Average error rate $p=0.6$ for sparse regime.}
		\label{fig:Noisy_p06}
	\end{subfigure}
	
	\begin{subfigure}{0.38\textwidth}
		\centering
		\captionsetup{justification=centering}
		\includegraphics[height=3.5cm, width = \textwidth]{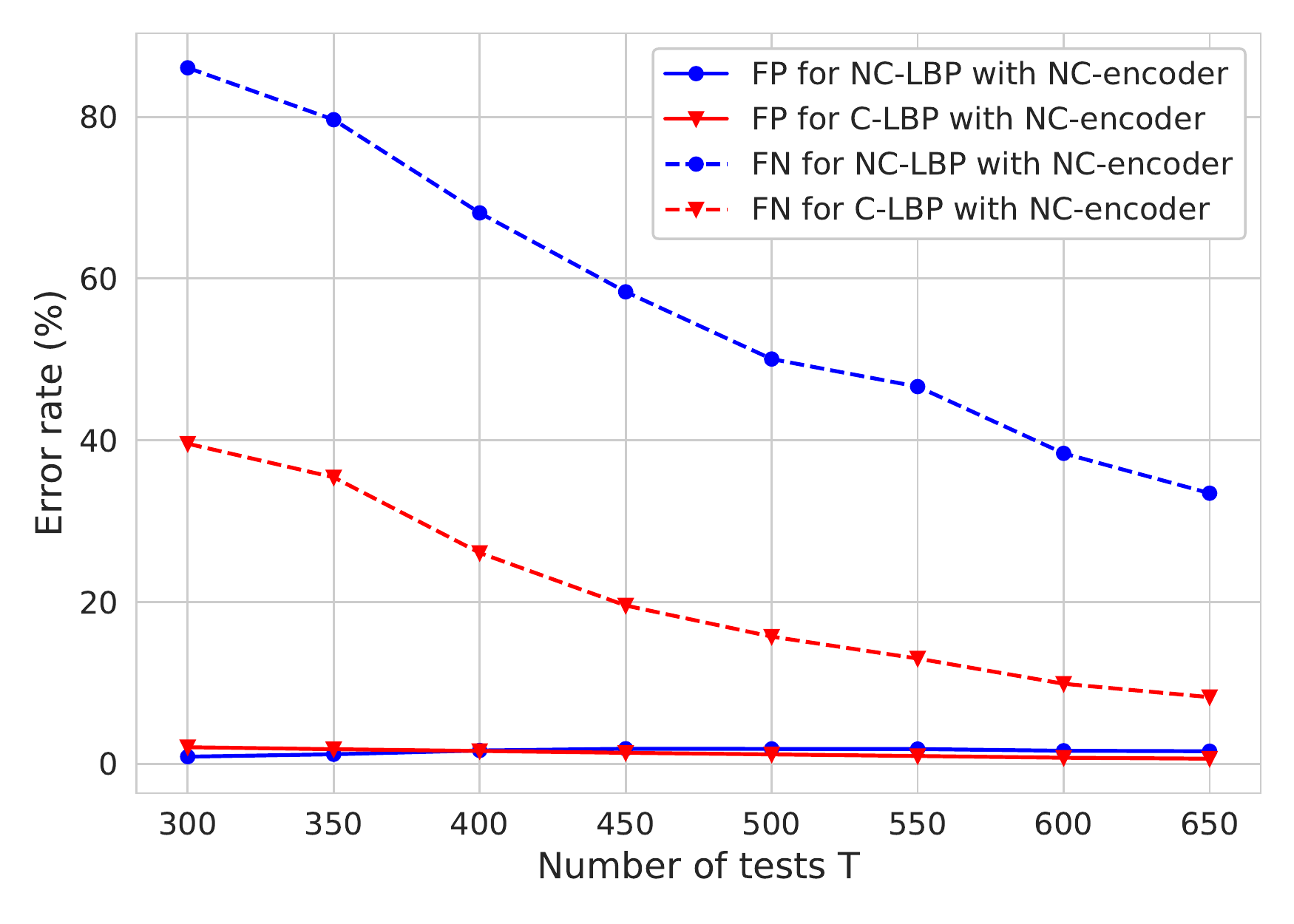}
		\caption{Noisy case: Average error rate $p=0.6$ for linear regime.}
		\label{fig:Noisy_linear_p06}
	\end{subfigure}
	
	\begin{subfigure}{0.38\textwidth}
		\centering
		\captionsetup{justification=centering}
		\includegraphics[height=3.5cm, width = \textwidth]{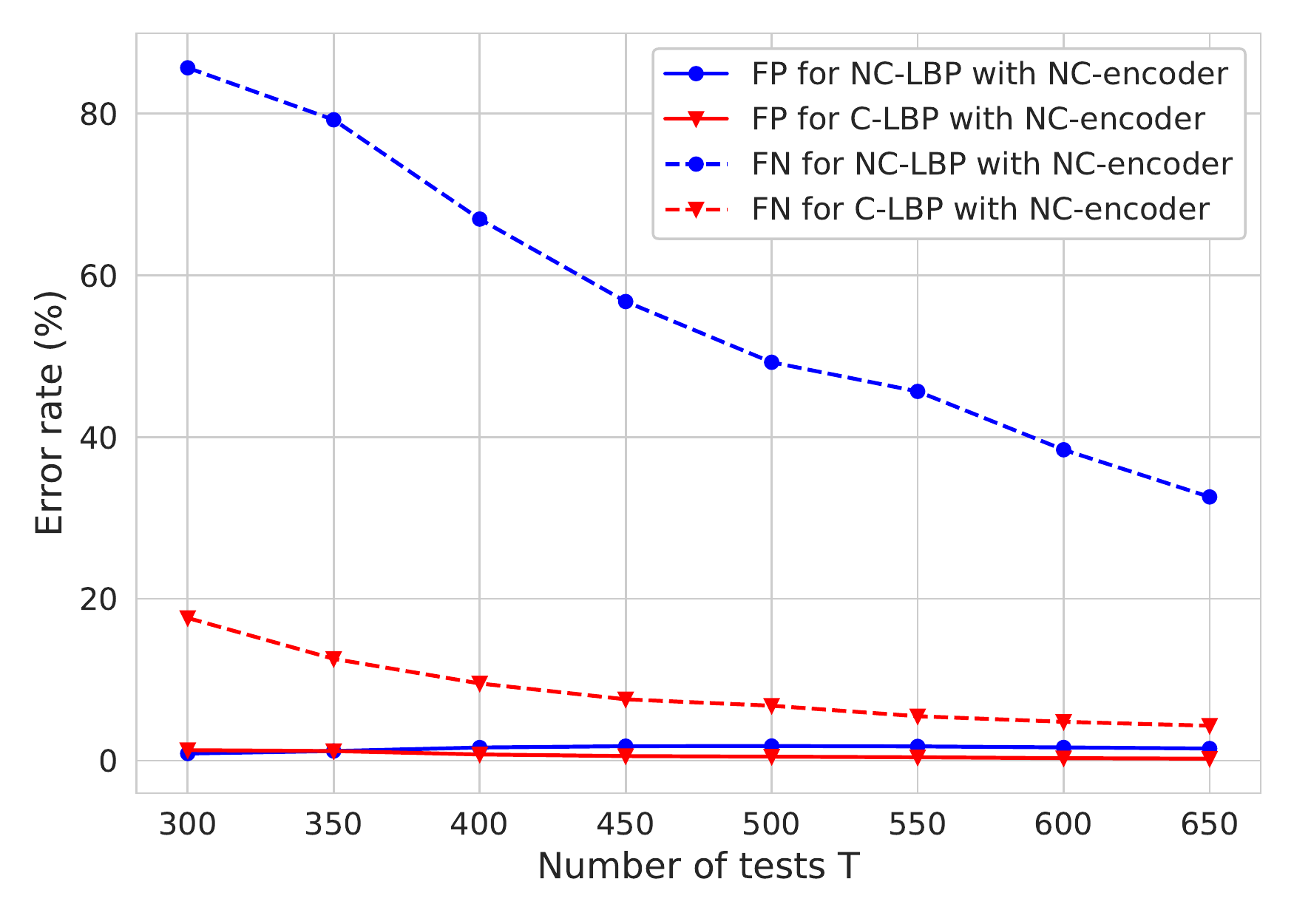}
		\caption{Noisy case: Average error rate $p=0.8$ for linear regime.}
		\label{fig:Noisy_linear_p08}
	\end{subfigure}
	\caption{Experiment (iii): \\Noisy case---Average error rate.}
	\label{fig:exp3}
\end{figure}

\begin{figure}
	\vspace{-0.2cm}
	\centering
	\captionsetup{justification=centering}	
	\includegraphics[ height=4cm,width=0.35\textwidth]{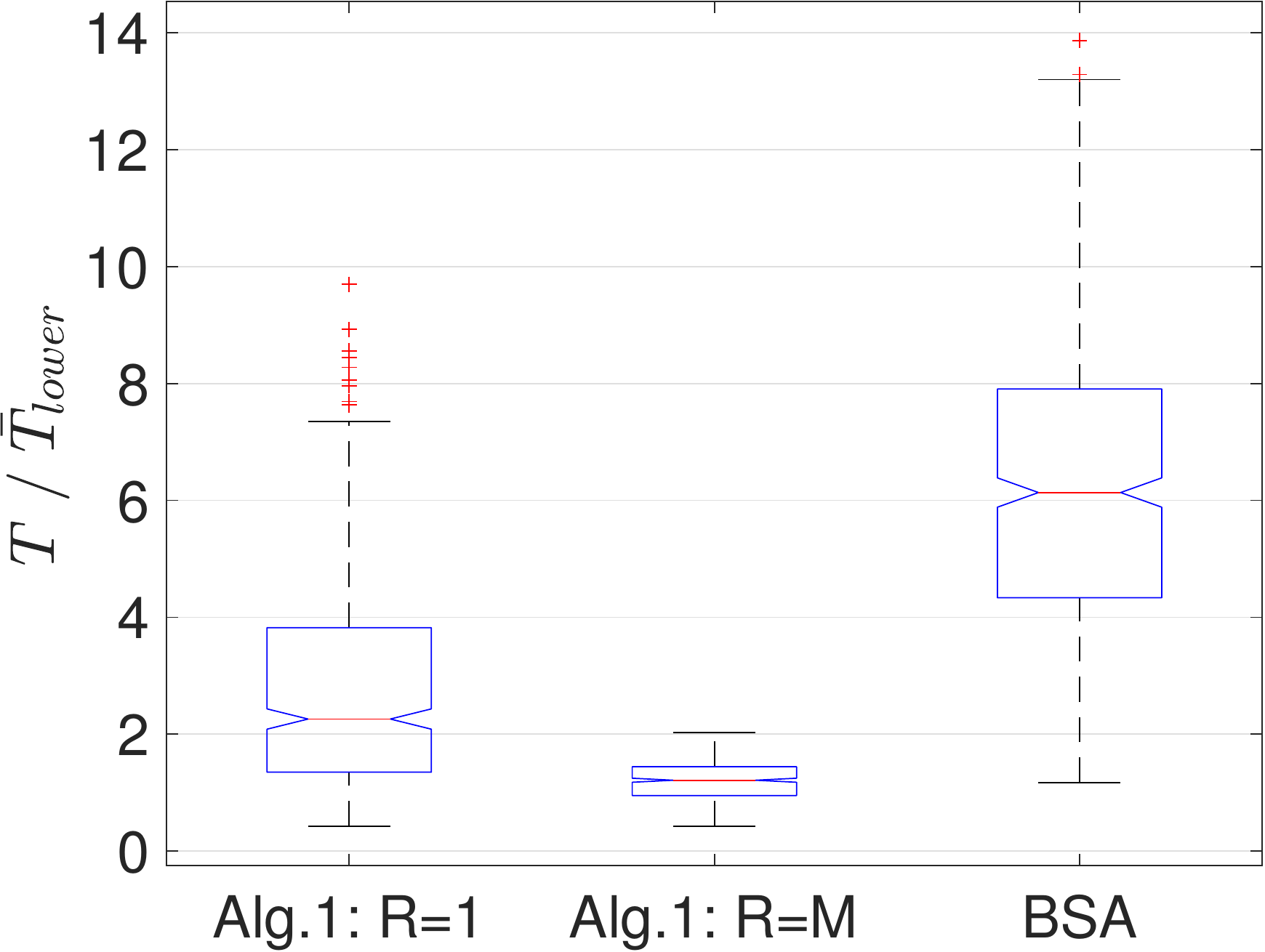}
	\vspace{-0.2cm}
	\caption{Asymmetric case---Linear regime: Cost efficiency for number of tests.}
	\label{fig:assymtericLinear}
	\vspace{-0.2cm}
\end{figure}

\end{document}